\documentclass[journal]{IEEEtran}
\usepackage{amsthm}
\usepackage{amsfonts}
\usepackage{amssymb}
\ifCLASSINFOpdf
  \usepackage[pdftex]{graphicx}
  \usepackage{graphicx}
  \graphicspath{{../pdf/}{../jpeg/}{../png/}}
  \DeclareGraphicsExtensions{.pdf,.jpeg,.png}
\else
  \usepackage[dvips]{graphicx}
  \graphicspath{{../eps/}}
  \DeclareGraphicsExtensions{.eps}
\fi
\ifCLASSOPTIONcompsoc
  \usepackage[caption=false,font=normalsize,labelfont=sf,textfont=sf]{subfig}
\else
\usepackage[caption=false,font=footnotesize]{subfig}

\usepackage{multicol}
\usepackage{epstopdf}
\usepackage{amsbsy}
\usepackage{float}
\newtheorem {thm}{Theorem}
\newtheorem{lem}{Lemma}
\usepackage{cite}
\ifCLASSINFOpdf
   \usepackage[pdftex]{graphicx}
   \else
\usepackage[dvips]{graphicx}
\fi
\usepackage[cmex10]{amsmath}
\usepackage{algorithm}
\usepackage{flushend}

\usepackage{relsize}
\usepackage{xcolor}
\usepackage{algpseudocode}
\usepackage[justification=centering]{caption}
\begin{document}
\title{Robust Linear Hybrid Beamforming Designs Relying on Imperfect CSI in mmWave MIMO IoT Networks}
\author{ Kunwar~Pritiraj~Rajput,~\IEEEmembership{Member,~IEEE,} Priyanka~Maity, Suraj~Srivastava,~\IEEEmembership{Member,~IEEE,} Vikas~Sharma,  Naveen~K.~D.~Venkategowda,~\IEEEmembership{Member,~IEEE,}     
        Aditya~K.~Jagannatham,~\IEEEmembership{Member,~IEEE,} \\ and Lajos Hanzo,~\IEEEmembership{Life Fellow,~IEEE}      
\thanks{L. Hanzo would like to acknowledge the financial support of the Engineering and Physical Sciences Research Council projects EP/W016605/1 and EP/P003990/1 (COALESCE) as well as of the European Research Council's Advanced Fellow Grant QuantCom (Grant No. 789028). The work of Aditya K. Jagannatham was supported in part by the
Qualcomm Innovation Fellowship, and in part by the Arun Kumar Chair
Professorship.}\thanks{K. P. Rajput,  P. Maity, S. Srivastava, and A. K. Jagannatham are with the Department of Electrical Engineering, Indian Institute of Technology, Kanpur, Kanpur, 208016, India (e-mail: \{pratiraj, pmaity, ssrivast, adityaj\}@iitk.ac.in).}
\thanks{V. Sharma is with Qualcomm India Pvt. Ltd., Hyderabad, India (e-mail: vikashar@qti.qualcomm.com)}
\thanks{N. K. D. Venkategowda is with the Department of Science and Technology, Linköping University,  60174 Norrköping, Sweden (e-mail: naveen.venkategowda@liu.se).}
\thanks{L. Hanzo is with the School of Electronics and Computer Science, University of Southampton, Southampton SO17 1BJ, U.K. (e-mail:lh@ecs.soton.ac.uk)}.}   
\maketitle
\vspace{-30pt}
\begin{abstract}
Linear hybrid beamformer designs are conceived for the decentralized estimation of a vector parameter in a millimeter wave (mmWave) multiple-input multiple-output (MIMO) Internet of Things network (IoTNe). The proposed designs incorporate both total IoTNe and individual IoTNo power constraints, while also eliminating the need for a baseband receiver combiner at the fusion center (FC). To circumvent the non-convexity of the hybrid beamformer design problem, the proposed approach initially determines the minimum mean square error (MMSE) digital transmit precoder (TPC) weights followed by a simultaneous orthogonal matching pursuit (SOMP)-based framework for obtaining the analog RF and digital baseband TPCs. Robust hybrid  beamformers are also derived for the realistic imperfect channel state information (CSI) scenario, utilizing both the stochastic and norm-ball CSI uncertainty frameworks. The centralized MMSE bound  derived in this work serves as a lower bound for the estimation performance of the proposed hybrid TPC designs. Finally, our simulation results quantify the benefits of the various designs developed.
\end{abstract}
\begin{keywords}
Coherent MAC, CSI uncertainty, linear decentralized estimation, hybrid beamforming design, mmWave communication, Internet of Things.
\end{keywords} 
\section{Introduction}
The Internet of Things (IoT) technology has compelling applications  pertaining to health care \cite{9365708,9055403}, smart  cities \cite{7959573},  surveillance \cite{8889708}, smart farming \cite{8521668}, amongst many others \cite{8879484}. IoT networks (IoTNes), wherein several miniature  IoT nodes (IoTNos) monitor  multiple phenomena of interest and relay suitably processed observations to a fusion center (FC), play a key role in the IoTNe \cite{9098849}. The challenging task of an IoTNe is to reliably estimate the underlying unknown parameter/ parameters from the noisy measurements collected from a large number of IoTNo deployed in the area of interest. However, it must be noted that each IoTNo has different constraints like finite battery life and limited computational power, which makes an IoTNe different from the cellular network. Therefore, one must aim to develop low-complexity algorithms for the design of the TPC at each IoTNe in the face of both power as well as MSE constraints, which forms a key focus of this work \cite{7879243}. Furthermore, the soaring data rates necessitated by the ever increasing gamut of data-hungry applications, such as virtual/ augmented reality, V2X communication and massive machine type communication (mMTC), are beginning to clog up the spectrum available  in the sub-6 GHz frequency band. As a remedy, millimeter wave (mmWave) has abundant bandwidth resources. However, it must be noted that communications at mmWave frequencies add new challenges to those in the sub-6 GHz bands, such as higher propagation losses and severe signal blockage \cite{heath2016overview,rappaport2014millimeter,8207426} arising due to the extremely low wavelength in the mmWave regime. Large antenna arrays alleviate the above problem via the transmission of focussed signal beams using beamforming. Thus, multiple-input multiple-output (MIMO) beamforming technology is an invaluable tool in the practical realization of mmWave communication.  In such systems, the conventional transceiver design scheme that employs an individual RF chain for each antenna is infeasible due to its excessive power consumption, cost and complexity. Therefore the hybrid MIMO transceiver architecture combining analogue phase-shifting based beamforming and a digital baseband beamformer \cite{el2012low,el2014spatially,rusu2015low,
yu2016alternating} was proposed to overcome this problem. This paves the way for the practical realization of mmWave MIMO systems. In this context, there is  an acute need to develop hybrid beamforming strategies for efficient signal processing in mmWave MIMO IoTNes. A comprehensive literature survey of the existing contributions on such systems is presented next. \par
 \begin{table*}[t]
\centering
 \caption{Boldly contrasting our novelty to the existing literature.}
\label{Literature}
\begin{tabular}{||l|c|c|c|c|c|c|c|c|c|c|c|c|c|c|}
\hline
Feature &[7] &[18] &[19] & [20]& [21]& [23] & [24]& [25]&[30]&[31] & \textbf{Our work} \\
\hline
\hline
Millimeter wave IoTNe & \checkmark & & & & & & & & & &  \pmb{\checkmark}\\
\hline
Coherent MAC & &\checkmark &\checkmark &\checkmark & \checkmark & \checkmark & \checkmark & \checkmark & & & \pmb{\checkmark}\\
\hline
Vector parameter estimation  &\checkmark &\checkmark &\checkmark & \checkmark &\checkmark & & & \checkmark& & &\pmb{\checkmark}\\
\hline
Hybrid precoder design  &\checkmark & & & & & & & &\checkmark &\checkmark  &\pmb{\checkmark}\\
\hline
Non-iterative solutions  & & & & \checkmark& &\checkmark & & & & &\pmb{\checkmark}\\
\hline
Total network power constraint  & & & & & &\checkmark & & & &  &\pmb{\checkmark}\\
\hline
Per IoTNo power constraint  & &\checkmark &\checkmark & &\checkmark & & & \checkmark & \checkmark &\checkmark   &\pmb{\checkmark}\\
\hline
Stochastic CSI uncertainty model  & & & & & & & & & \checkmark & \checkmark &\pmb{\checkmark}\\
\hline
Norm ball CSI uncertainty model & & & & & &\checkmark & & & & & \pmb{\checkmark}\\
\hline
\end{tabular}
\end{table*}  
\subsection{Review of existing literature}
In a typical linear decentralized estimation setup, a large number of IoTNos are deployed in an area of interest which one wants to monitor. The IoTNos transmit their precoded observations over a coherent multiple access channel (MAC) to the FC for receiver combining (RC) and final estimation \cite{9464966}. Coherent MAC requires that all the IoTNos in the network are synchronized and their individual transmitted signals arrive as a coherent sum at the FC. The FC determines the transmit precoders (TPCs) for all IoTNos in the IoTNe and subsequently, feeds back to each IoTNo its own TPC weights through a feedback link. The pioneering paper by Xiao \textit{et al.} \cite{xiao2008linear} proposed an optimal linear transceiver design for the efficient
decentralized estimation of a vector parameter for a coherent MAC-based MIMO IoTNe. However, the analysis of \cite{xiao2008linear} assumed a simplistic diagonal model for the MIMO channel between each IoTNo and the FC, which restricts the applicability of their results. By extending this framework, Behbahani \textit{et al.} \cite{behbahani2012linear} proposed a ground-breaking iterative procedure for a joint transceiver design relying on mean square error (MSE) minimization  for vector parameter estimation at the FC. Along similar lines, Liu \textit{et al.} \cite{liu2017joint} developed decentralized and distributed schemes for the efficient estimation of a parameter vector. However, the algorithms developed in \cite{behbahani2012linear} and \cite{liu2017joint} are based on the iterative block coordinate descent (BCD) framework, which results in a high computational complexity, especially in a MIMO IoTNe associated with a large number of IoTNos. The authors of \cite{6470681,7572989,7060716} proposed an optimal IoTNo collaboration strategy, wherein each IoTNo shares its observations with its neighbouring IoTNos and only a subset of the IoTNos transmit their observations to the FC for the final estimation of the unknown parameter, which leads to an improved estimation performance. However, the  inter-IoTNo communication required for supporting this can potentially impose high overheads. \par
Another major drawback of all the above contributions is their reliance on perfect channel state information (CSI), which can be a significant stumbling block in the practical implementation of such systems. Given the stringent bandwidth constraints, it is also desirable that each node transmits very few pilot symbols to the FC, which often leads to a very coarse estimate of the channel state information (CSI). Therefore, this forms the motivation to develop robust hybrid TPC designs incorporating the uncertainty in the available estimate of the CSI. To address this limitation of the existing literature, Venkategowda \textit{et al.} \cite{venkategowda2017precoding} developed robust TPC designs by considering both the so-called norm ball and the ellipsoidal CSI uncertainty models. By further extending this idea,  Liu \textit{et al.}\cite{8660565} developed centralized and distributed robust transceiver designs by relying on MSE minimization subject to per IoTNo power constraints, and total power minimization under a specific MSE constraint. However, their work considered only the ellipsoidal CSI uncertainty model. Iterative schemes were also conceived for robust transceiver design for vector parameter estimation in \cite{7467438} by considering the norm ball CSI uncertainty model. At this juncture, it is important to note that while the above treatises have comprehensively addressed the problem of TPC/ RC design for IoTNes, they cannot be directly  extended to a 5G mmWave-based MIMO IoTNe relying on hybrid  beamforming. Hence, the contributions specifically related to mmWave MIMO beamforming are reviewed in detail next.\par
Hybrid TPC/ RC designs have been proposed in\cite{el2012low,el2014spatially,rusu2015low,
yu2016alternating, 9395091,9445013} for cellular mmWave MIMO systems. The authors of \cite{9410457} developed a novel signal processing strategy for the reconfigurable intelligent surface (RIS)-aided MIMO uplink for energy and spectral efficiency maximization. An interesting analysis has been carried out therein to study the trade off between the energy and spectral efficiency. Furthermore, Huang et al.
\cite{9309152} have proposed an interesting deep reinforcement learning-based hybrid beamforming design for an RIS-aided Terahertz system. \textcolor{black}{Recently, the authors of \cite{liu2021hybrid} have proposed both the centralized (C)- and the asynchronous distributed (AD)-alternating direction method of multipliers (ADMM)-based schemes for linear minimum MSE (LMMSE) hybrid transceivers employed in coherent MAC-based IoTNe. Furthermore, Liu \textit{et al.} \cite{9098849} proposed a sophisticated technique for hybrid transceiver designs relying on perfect CSI in the context of an orthogonal MAC-based IoTNe. Briefly, hybrid transceiver designs were developed relying on the ADMM and steepest descent (SD) techniques. Since, the designs proposed in \cite{liu2021hybrid} and \cite{9098849} are iterative, they potentially lead to a high complexity. Additionally, it must also be borne in mind that the orthogonal MAC employed in \cite{9098849} for IoTNo communication has a low bandwidth efficiency in comparison to the coherent MAC \cite{liu2021hybrid,xiao2008linear}, especially in modern IoTNes having a large number of nodes.} Similar to sub-6 GHz systems, the availability of perfect CSI is also a serious limiting factor in the design of mmWave MIMO systems, which has only been addressed by a few studies \cite{8537897,8906245,9079551}. Luo \textit{et al.} \cite{8537897} conceived a hybrid transceiver for a scenario having a single source and multiple relays. The proposed designs target the maximization of the average signal-to-noise (SNR) power ratio at the destination using realistic imperfect CSI of the source-relay and relay-destination links. The cutting-edge paper of Jiang and Jafarkhani  \cite{8906245} designed a low complexity robust hybrid transceiver for a single-relay cooperative communication system, leveraging the principle of mutual information maximization. As a further advance, Zhao \textit{et al.} \cite{9079551} proposed a scheme for robust joint hybrid transceiver design in a full-duplex (FD) multi-cell mmWave network to maximize the sum-rate. However, the majority of these solutions cannot be incorporated in a IoTNe framework, since they do not address the associated stringent bandwidth and power limitations.  Furthermore, the authors of \cite{8537897,8906245,9079551}  consider only the stochastic CSI uncertainty model, which can lead to a poor worst-case performance. To the best of our knowledge, none of the existing treatises derived the hybrid beamforming techniques for vector parameter estimation in a coherent MAC- based mmWave MIMO IoTNe, accounting also for the imperfect nature of the CSI in a realistic implementation. Hence we are inspired to fill this knowledge gap in the existing literature. The main contributions of this study are itemized next and they are also boldly contrasted to the literature in Table-I.
\begin{figure*}[!t]


\normalsize









\begin{center}
\includegraphics[width=0.8\linewidth]{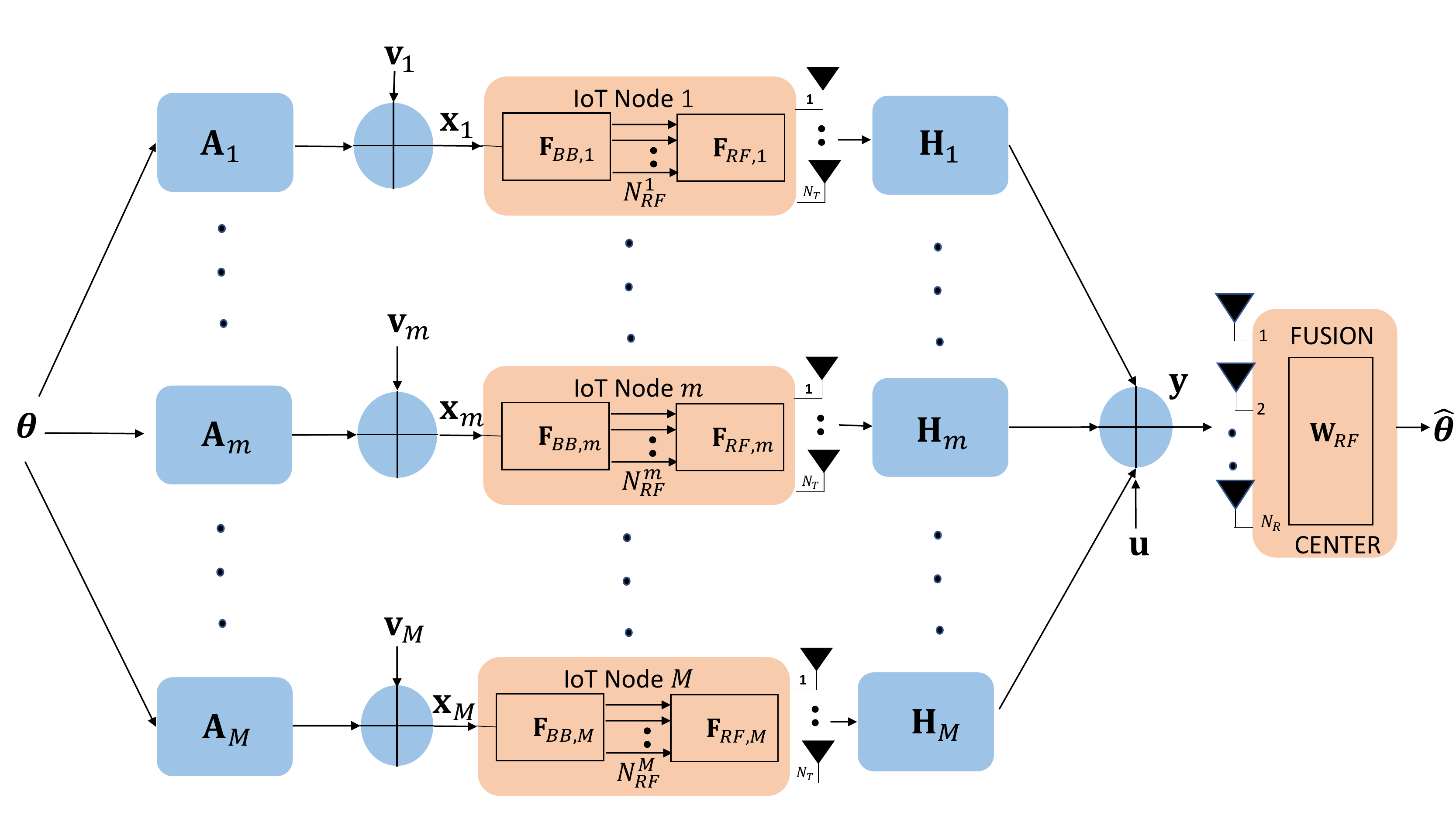}
\end{center}
\caption{System model depicting hybrid signal processing for the mmWave MIMO IoT network.}
\label{fig:noisy varing sensor}

\hrulefill



\vspace{-10pt}

\end{figure*}
\subsection{Contributions of this Treatise} 
\begin{itemize}
\item Non-iterative hybrid linear beamforming techniques are proposed for minimizing the MSE of vector parameter estimation at the FC in a mmWave MIMO IoTNe, initially assuming perfect CSI availability. Sophisticated optimization paradigms are employed for incorporating both total and per-IoTNo power constraints into our TPC design.
\item A two-step procedure is proposed for overcoming the non-convexity associated with the MSE optimization problem, arising due to the constant gain constraint for the elements of the RF TPC and RC. The optimality of this procedure is rigorously justified via a mathematical proof of the proposition. Closed-form expressions are derived for the fully digital minimum MSE (MMSE) TPCs, employing a suitable setting for the RF combiner.
\item Subsequently, a low-complexity simultaneous orthogonal matching pursuit (SOMP)-based technique is proposed for designing the RF and baseband components of the fully-digital MMSE precoder.
\item In contrast to the existing works \cite{9098849}, \cite{liu2021hybrid}, we consider the realistic scenario of practical CSI uncertainty, and subsequently develop robust linear hybrid TPCs by modelling the CSI uncertainty using both the stochastic and the norm ball CSI uncertainty models for the above system. This reduces the pilot overhead and allows the IoTNe to work well even when the channel estimate is not very accurate.
\item The centralized MMSE bound is derived, which acts as a benchmark for characterizing the estimation performance of the proposed designs. 
\item Finally, our simulation results illustrate the performance of the various analytical solutions derived in this paper.
\end{itemize}\par
\subsection{Organization of the paper}
The remainder of the paper is organized as follows. Section-\ref{sys_mod} presents the mmWave MIMO IoTNe system model and formulates the MSE minimization problem for parameter estimation. Section-III describes the detailed procedure of obtaining the fully digital MMSE TPCs and the algorithm of determining the RF/ baseband TPC weights. Section-\ref{Imperfect_CSI} describes a pair of novel robust hybrid pre-processing designs accounting for the CSI uncertainty using both the stochastic and the norm ball CSI uncertainty models. Our simulation results are provided in Section-\ref{sim_results}, while our conclusions are given in Section-\ref{Con}. The proofs of the various propositions are provided in the Appendices. \par
\textbf{Notation:} Small and capital boldface letters $\mathbf{d}$ and $\mathbf{D}$ denote vectors and matrices, respectively. The transpose and Hermitian transpose operations are denoted by $(.)^T$ and $(.)^H$, respectively. Furthermore, $\left[\mathbf{D}\right]_{i,j}$ denotes the $(i,j)$th element of a matrix $\mathbf{D}$. The operator $\mathbb{E}\{.\}$ denotes the statistical expectation operator, while $\mathrm{diag}(\mathbf{D}_1,\mathbf{D}_2,\hdots,\mathbf{D}_L)$ denotes a block diagonal matrix constituted by the matrices $\mathbf{D}_1$, $\mathbf{D}_2$ upto $\mathbf{D}_L$ on its principal diagonal. The $l_2$ and Frobenius norms of vectors and matrices, respectively, are represented by $||.||$ and $||.||_F$, respectively, whereas $l_0$ norm is represented by $||.||_0$. The symbol $\otimes$ denotes the matrix Kronecker product and $\lambda_{\text{max}}(\mathbf{D})$ represents the maximum eigenvalue of the matrix $\mathbf{D}$. The operator $\mathrm{vec}(\mathbf{D})$ stacks the columns of a matrix $\mathbf{D}$ into a vector, while the operator $\mathrm{vec}^{-1}_m(\mathbf{d})$ converts a vector $\mathbf{d}$ into a matrix of $m$ rows and the appropriate number of columns. $\mathbf{X}=\mathbf{D}\left(:,\mathcal{K}\right)$ represents a sub-matrix $\mathbf{X}$ consisting of those columns of the matrix $\mathbf{D}$ whose information is provided by the set $\mathcal{K}$.
\section{mmWave IoTNe System Model and Problem Formulation}\label{sys_mod}
Similar to \cite{9098849}, \cite{liu2021hybrid}, \cite{9372293}, we consider the mmWave-based MIMO IoTNe depicted in Fig. \ref{fig:noisy varing sensor}, where $M$ multi-antenna aided IoTNos are deployed to observe/ sense an unknown vector parameter of interest denoted by $\boldsymbol{\theta}\in \mathbb{C}^{p \times 1}$ that is assumed to be distributed with mean of zero and covariance matrix of $\mathbf{R}_{\theta} \in \mathbb{C}^{p \times p}$. For instance, in an environmental monitoring application, the elements of the $p$ dimensional parameter $\boldsymbol{\theta}$ can represent various physical quantities, like pressure, temperature, moisture content etc.. The observation vector corresponding to IoTNo $m$, denoted by $\mathbf{x}_m\in \mathbb{C}^{q \times 1}$, and can be modeled as\cite{xiao2008linear, behbahani2012linear,9461735}
\begin{equation}
    \mathbf{x}_m = \mathbf{A}_m\boldsymbol{\theta} + \mathbf{v}_m,
\end{equation}
where $\mathbf{A}_m \in \mathbb{C}^{q \times p}$ denotes the observation matrix, while $\mathbf{v}_m \in \mathbb{C}^{q\times 1}$ represents the circularly symmetric complex Gaussian observation noise with mean zero and covariance matrix $\mathbf{R}_{m} \in \mathbb{C}^{q\times q}$. The quantity ${q}$ represents the number of measurements taken by the $m$th IoTNo. Subsequently, the observation vector $\mathbf{x}_m$ is precoded using the baseband TPC $\mathbf{F}_{\text{BB},m} \in \mathbb{C}^{{N}^m_{\text{RF}} \times q}$, followed by an RF TPC  $\mathbf{F}_{\text{RF},m} \in \mathbb{C}^{{N}_T \times {N}^m_{\text{RF}}}$, where the quantities $N_T$ and ${N}^m_{\text{RF}}$ denote the number of antennas and RF chains employed at the $m$th sensor, respectively. The pre-processed observation vectors corresponding to each IoTNo $m$ are subsequently transmitted over a coherent MAC to the FC. Hence, the signal $\mathbf {y} \in \mathbb{C}^{N_R \times 1}$ received by the FC can be mathematically expressed as
\begin{figure}
\centering 
\includegraphics[width=\linewidth]{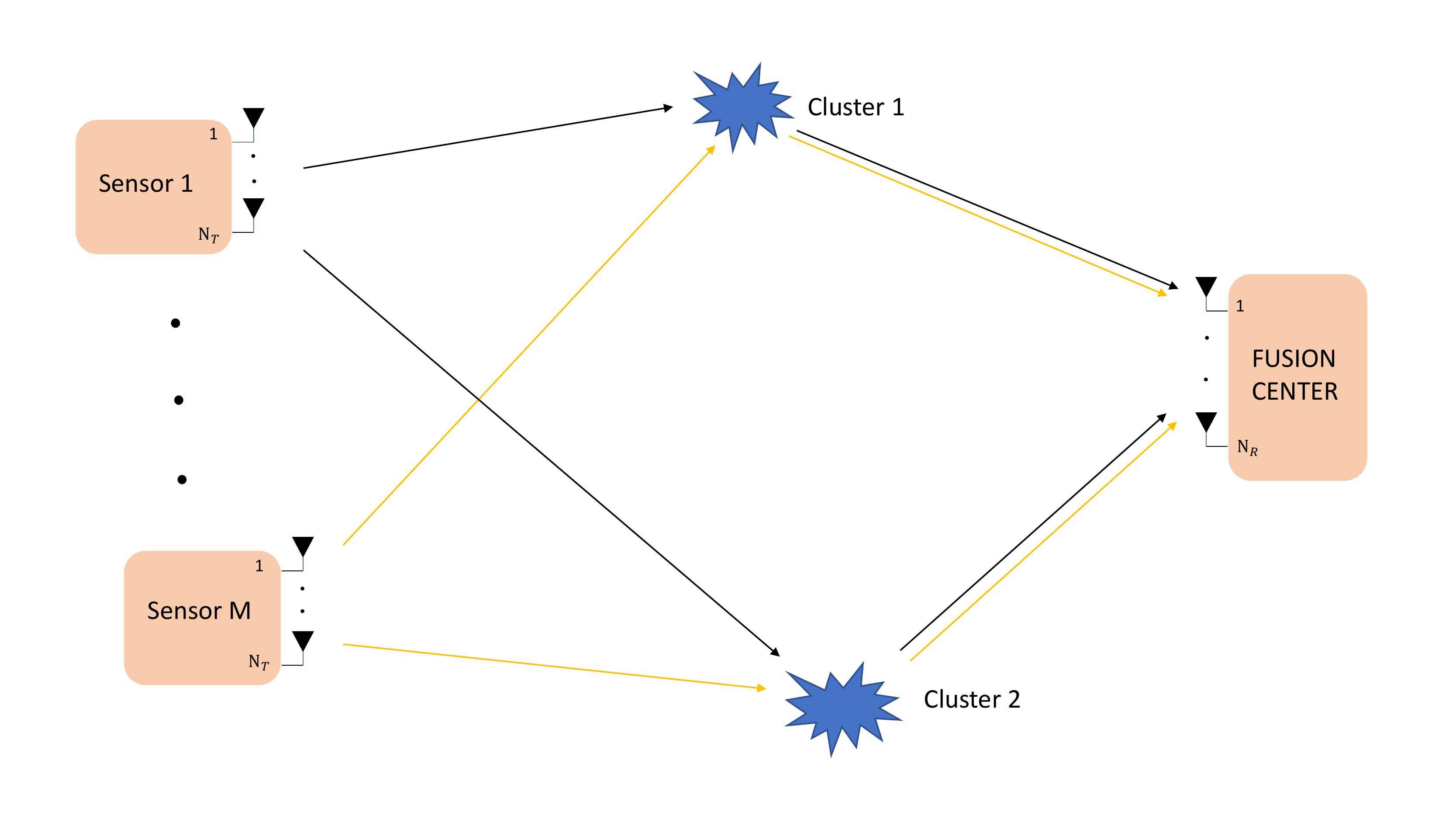}
\caption{\textcolor{black}{Spatial channel model of the mmWave MIMO channel}}\label{SCM}
\end{figure}
\begin{align}\label{Received_vec}
&\mathbf{y}= \sum_{m = 1}^{M} \mathbf{H}_m \mathbf{F}_{\text{RF},m} \mathbf{F}_{\text{BB},m} \mathbf{x}_m + \mathbf{u}\nonumber \\
&= \sum_{m = 1}^{M} \mathbf{H}_m \mathbf{F}_{\text{RF},m} \mathbf{F}_{\text{BB},m} \mathbf{A}_m \boldsymbol{\theta} + \sum_{m = 1}^{M} \mathbf{H}_m \mathbf{F}_{\text{RF},m} \mathbf{F}_{\text{BB},m} \mathbf{v}_m + \mathbf{u},
\end{align}
where $\mathbf{H}_{m} \in \mathbb{C}^{{N}_R \times {N}_{T}}$ denotes the mmWave MIMO channel between the $m$th IoTNo and the FC, which is equipped with $N_R$ antennas. The vector $\mathbf{u} \in \mathbb{C}^{{N}_R \times 1}$ denotes the circularly symmetric complex Gaussian noise vector at the FC, distributed with mean $\mathbf{0}$ and variance $\mathbf{R}_u \in \mathbb{C}^{{N}_R \times {N}_R}$. Subsequently, an RF RC $\mathbf{W}_{\text{RF}} \in \mathbb{C}^{{N}_R \times {N}^{\text{FC}}_{\text{RF}}}$ is employed at the FC. The resultant signal $\widetilde{\mathbf{y}} \in \mathbb{C}^{{N}^{\text{FC}}_{\text{RF}} \times 1}$ at the output of the RF RC is given by
\begin{align}\label{Rec_y}
\widetilde{\mathbf{y}}&=\mathbf{W}_{\text{RF}}^H \mathbf{y}=\sum_{m = 1}^{M} \mathbf{W}_{\text{RF}}^H \mathbf{H}_m \mathbf{F}_{\text{RF},m} \mathbf{F}_{\text{BB},m}  \mathbf{A}_m\boldsymbol{\theta} \nonumber 
\\ &+ \sum_{m = 1}^{M} \mathbf{W}_{\text{RF}}^H \mathbf{H}_m \mathbf{F}_{\text{RF},m} \mathbf{F}_{\text{BB},m} \mathbf{v}_m +{\mathbf{W}_{\text{RF}}^H\mathbf{u}},
 \end{align}
where ${N}^{\text{FC}}_{\text{RF}}$ denotes the number of RF chains employed at the FC. \textcolor{black}{Using the spatial channel model depicted in Fig.\ref{SCM}, the mmWave MIMO channel can be modeled as follows.}  As described in \cite{tse2005fundamentals,heath2016overview,srivastava2019quasi}, the channel matrix $\mathbf {H}_m$ between the FC and the $m$th IoTNo is given by
\begin{align}\label{mmWave_Ch}
\mathbf{H}_m=\sqrt{\displaystyle\frac{N_T N_R}{N_{cl}}}\sum\limits_{k=1}^{N_{cl}} \alpha^m_{k} \mathbf{a}_{\text{fc}}(\phi_{k}) \mathbf{a}^H_{\text{s}}(\theta_{k,m}),
\end{align}
where the 3-tuple $(\alpha^m_{k}, \phi_k, \theta_{k,m})$ represents the complex path gain $\alpha^m_{k}$, angle of arrival (AoA) $\phi_k$ at the FC, and angle of departure (AoD) $\theta_{k,m}$ corresponding to the $m$th IoTNo associated with the $k$th cluster, and $N_{cl}$ denotes the total number of clusters. The vectors $\mathbf{a}_{\text{fc}}(\phi_k)\in \mathbb{C}^{N_R \times 1}$ and $\mathbf{a}_{\text{s}}(\theta_{k,m})\in \mathbb{C}^{N_T \times 1}$ denote the array response vectors at the FC and the $m$th sensor, respectively, corresponding to the $k$th cluster, which are defined as
\begin{equation}
\mathbf{a}_{\text{fc}}(\phi_k)\hspace{-2pt}=\hspace{-2pt}\displaystyle\frac{1}{\sqrt{N_R}} \Big[ 1, e^{-j\tilde{\phi}_k}, \cdots , e^{-j(N_R-1)\tilde{\phi}_k}\Big]^T\hspace{-5pt}, \label{eq: rxarrresvec}
\end{equation}
\begin{equation}
\mathbf{a}_{\text{s}}(\theta_{k,m})\hspace{-2pt}=\hspace{-2pt}\displaystyle\frac{1}{\sqrt{N_T}} \Big[ 1, e^{-j\tilde{\theta}_{k,m}}, \cdots , e^{-j(N_T-1)\tilde{\theta}_{k,m}}\Big]^T\hspace{-4pt}, \label{eq: txarrresvec}
\end{equation}
where $\tilde{\phi}_k = \frac{2 \pi}{\lambda}d_R \cos(\phi_k)$ and $\tilde{\theta}_{k,m}= \frac{2 \pi}{\lambda}d_T \cos(\theta_{k,m})$. The quantities $\lambda, d_R,$ and $d_T$ denote the carrier's wavelength, as well as the inter-antenna spacings at the FC and each sensor, respectively. The mmWave MIMO channel $\mathbf{H}_m$ in \eqref{mmWave_Ch} can be equivalently represented in the compact form 
\begin{align}
\mathbf H_m = \mathbf{A}_{\text{fc}} \mathbf{D}_m \mathbf{A}_{\text{s},m}^{H},
\end{align}
where $\mathbf{A}_{\text{fc}} = [\mathbf a_{\text{fc}}(\phi_1), \cdots  , \mathbf a_{\text{fc}}(\phi_{N_{cl}})] \in \mathbb{C}^{N_R \times N_{cl}} , \ \mathbf{A}_{\text{s},m} = [\mathbf a_{\text{s}}(\theta_{1,m}), \cdots  , \mathbf a_{\text{s}}(\theta_{N_{cl},m})] \in \mathbb{C}^{N_T \times N_{cl}}$ are termed the array response matrices corresponding to the FC and the $m$th sensor, respectively, and the diagonal matrix $\mathbf{D}_m \in \mathbb{C}^{N_{cl} \times N_{cl}}$ is defined as  $\mathbf{D}_m=\sqrt{\frac{N_T N_R}{N_{cl}}} \mathrm{diag}(\alpha^m_1, \cdots, \alpha^m_{N_{cl}})$. The next section proposes hybrid linear beamforming designs for minimizing the MSE of parameter estimation at the FC.
\section{Linear hybrid beamforming designs with perfect CSI} \label{perfect_CSI}
Note that the signal $\widetilde{\mathbf{y}}$ in \eqref{Rec_y} represents the estimate of the unknown parameter vector $\boldsymbol{\theta}$ if the following condition holds:
\begin{align} \label{eq1}
\sum_{m=1}^{M}\mathbf{W}_{\text{RF}}^H \mathbf{H}_m \mathbf{F}_{\text{RF},m} \mathbf{F}_{\text{BB},m}  \mathbf{A}_m=\mathbf{D}=\begin{bmatrix}
\mathbf{I}_p\\
\mathbf{0}_{\left(N_{\text{RF}}^{\text{FC}}-p\right)\times p}
\end{bmatrix}, 
\end{align} 
where $\mathbf{D} \in \mathbb{C}^{N_{\text{RF}}^{\text{FC}} \times p}$. Due to the above estimation constraint, the number of RF chains $N_{\text{RF}}^{\text{FC}}$ at the FC must be greater than or equal to the number of parameters $p$, i.e., $N_{\text{RF}}^{\text{FC}} \geq p$. However, to limit the number of RF chains, we set $N_{\text{RF}}^{\text{FC}}=p$, i.e., $\mathbf{D}=\mathbf{I}_p$. The resultant MSE expression can be derived as 
\begin{align}\label{MSE}
\mathrm{MSE}=&\sum_{m = 1}^{M}\mathrm{Tr}\Big[\mathbf{W}_{\text{RF}}^H \mathbf{H}_m \mathbf{F}_{\text{RF},m} \mathbf{F}_{\text{BB},m} \mathbf{R}_m\mathbf{F}_{\text{BB},m}^H \mathbf{F}^H_{\text{RF},m}  \nonumber \\
&\times \mathbf{H}^H_m \mathbf{W}_{\text{RF}}\Big]+\mathrm{Tr}\left[\mathbf{W}_{\text{RF}}^H\mathbf{R}_u\mathbf{W}_{\text{RF}}\right].
\end{align}
The general optimization problem of minimizing the MSE in \eqref{MSE} obeying the zero-forcing constraint in \eqref{eq1} is given by
 \begin{eqnarray}\label{eq3}
 \begin{aligned}
 \underset{\mathbf{F}_{\text{RF},m},\mathbf{F}_{\text{BB},m},\newline \mathbf{W}_{\text{RF}}}{\text{minimize}}& \mathrm{MSE}  \\
\text{subject to} &  \sum\limits_{m=1}^{M}\mathbf{W}_{\text{RF}}^H \mathbf{H}_m \mathbf{F}_{\text{RF},m} \mathbf{F}_{\text{BB},m}  \mathbf{A}_m=\mathbf{I}_{p},\\
 &  \left\lvert\left[\mathbf{F}_{\text{RF},m}\right]_{i,j}\right\rvert=\frac{1}{\sqrt{N_T}}, \; 1 \leq m \leq M, \\
 & \left\lvert\left[\mathbf{W}_{\text{RF}}\right]_{i,j}\right\rvert=\frac{1}{\sqrt{N_R}}.
 \end{aligned}
 \end{eqnarray}
\textcolor{black}{The optimization problem in \eqref{eq3} is non-convex since both the objective function and the first constraint are non-convex in terms of the optimization variables $\left\{\{\mathbf{F}_{\text{RF},m}\}_{m=1}^M,\{\mathbf{F}_{\text{BB},m}\}_{m=1}^M,\mathbf{W}_{\text{RF}}\right\}$,  which are coupled with each other.} In addition, the constant magnitude constraints on the elements of the RF TPCs $\left\{\mathbf{F}_{\text{RF},m}\right\}_{m=1}^{M}$ and combiner $\mathbf{W}_{\text{RF}}$ are also non-convex. Hence, the optimization problem in  \eqref{eq3} is intractable. Our novel procedure to be described next overcomes this obstacle of designing the optimal TPCs and RC. To begin with, an appropriate choice of the RF RC $\mathbf{W}_{\text{RF}}$  is obtained as follows.\par 
Let the quantity $\alpha_k = \sum_{m=1}^{M} \vert \alpha^m_{k} \vert$ denote the sum of the magnitudes of complex path-gains of each IoTNo corresponding to the $k$th cluster. One can now arrange  the quantities $\alpha_{k}$ in decreasing order of their magnitudes as $\vert\alpha_{k_1}\vert \geq \vert\alpha_{k_2}\vert \geq  \cdots \geq \vert\alpha_{k_{N_{cl}}}\vert$. The RF combiner $\mathbf{W}_{\text{RF}}$ can now be designed as
\begin{align}
\mathbf{W}_{\text{RF}}& = \mathbf{A}_{\text{fc}}\left( \ :\ ,\  \mathcal{K}\ \right)=\left[\mathbf{a}_{\text{fc}}(\phi_{k_1}), \mathbf{a}_{\text{fc}}(\phi_{k_2}),\hdots,\mathbf{a}_{\text{fc}}(\phi_{k_{N^{\text{FC}}_{\text{RF}}}})\right],
\end{align}
where $\mathcal{K} = \left\{ k_1, k_2, \hdots, k_{N^{\text{FC}}_{\text{RF}}} \right\}$. Upon employing this setting of the RF RC $\mathbf{W}_{\text{RF}}$, the optimal fully digital TPC $\mathbf{F}_{m}$ for each IoTNo $m$ can be designed as described next. \par
To this end, employing the following identity  \cite{zhang2017matrix}
\begin{align}\label{Trace_vec}
\mathrm{Tr}\left[\mathbf{L}^H \mathbf{N}\mathbf{P}\mathbf{Q}\right]=\mathrm{vec}(\mathbf{L})^H\left(\mathbf{Q}^T \otimes \mathbf{N}\right)\mathrm{vec}(\mathbf{P}),
\end{align}
and further defining the fully-digital TPC matrix and vector  corresponding to the TPC vector of the $m$th IoTNo as $\mathbf{F}_m=\mathbf{F}_{\text{RF},m}\mathbf{F}_{\text{BB},m}$, and $\mathbf{f}_m=\mathrm{vec}(\mathbf{F}_m) \in \mathbb{C}^{qN_T \times 1}$, respectively, the MSE expression in \eqref{MSE} can be rewritten as 
\begin{align}
\mathrm{MSE}&=\sum_{m=1}^{M}\mathbf{f}^H_m\boldsymbol{\Psi}_m\mathbf{f}_m+\mathrm{Tr}\left[\mathbf{W}^H_{\text{RF}}\mathbf{R}_u\mathbf{W}_{\text{RF}}\right] \nonumber \\
&=\mathbf{f}^H\boldsymbol{\Psi}\mathbf{f}+\mathrm{Tr}\left[\mathbf{W}^H_{\text{RF}}\mathbf{R}_u\mathbf{W}_{\text{RF}}\right], \label{MSE1}
\end{align}
where the matrix $\mathbf{\Psi}_m \in \mathbb{C}^{qN_T \times qN_T}$ is defined as $\mathbf{\Psi}_m=\left(\mathbf{R}_m^T \otimes \mathbf{H}^H_m \mathbf{W}_{\text{RF}} \mathbf{W}^H_{\text{RF}}\mathbf{H}_m\right)$. Furthermore, the quantities $\mathbf{f} \in \mathbb{C}^{MqN_T \times 1}$ and $\mathbf{\Psi} \in \mathbb{C}^{MqN_T \times MqN_T}$ are defined as 
\begin{align}
\mathbf{f}&=\left[\mathbf{f}_1^H,\mathbf{f}_2^H,\hdots,\mathbf{f}_M^H\right]^H,  \\
\mathbf{\Psi}&=\mathrm{diag}\left(\mathbf{\Psi}_1,\mathbf{\Psi}_2,\hdots,\mathbf{\Psi}_M\right).
\end{align}
Upon applying the vectorization and identity $\mathrm{vec}\left(\mathbf{N}\mathbf{P}\mathbf{Q}\right)=\left(\mathbf{Q}^T \otimes \mathbf{N}\right)\mathrm{vec}(\mathbf{P})$, on both sides of the constraint in \eqref{eq1}, the ZF constraint can be written as
\begin{align}\label{eq5}
\sum_{m=1}^{M} \mathbf{Z}_m\mathbf{f}_m=\mathbf{Z}\mathbf{f}=\mathbf{b}, 
\end{align}
where the matrices obey $\mathbf{Z}_m=\left[\mathbf{A}^T_m \otimes \mathbf{W}_{\text{RF}}^H \mathbf{H}_m\right] \in \mathbb{C}^{p^2 \times qN_T}$ and $\mathbf{Z}= \left[\mathbf{Z}_1,\mathbf{Z}_2,\hdots,\mathbf{Z}_M\right] \in \mathbb{C}^{p^2 \times MqN_T}$, and we have $\mathbf{b}=\mathrm{vec}\left( \mathbf{I}_{p}\right) \in \mathbb{C}^{p^2 \times 1}$. Hence, the pertinent optimization problem of minimizing the resultant MSE in \eqref{MSE1} subject to the estimation constraint in \eqref{eq5} is given by
 \begin{eqnarray}\label{MSE_Min}
 \underset{\mathbf{f}}{\text{minimize}}& \mathbf{f}^H\mathbf{\Psi}\mathbf{f} \nonumber\\
 \text{subject to} & \mathbf{Z}\mathbf{f} = {\mathbf{b}},
 \end{eqnarray}
where the constant term $\mathrm{Tr}\left[\mathbf{W}^H_{\text{RF}}\mathbf{R}_u\mathbf{W}_{\text{RF}}\right]$ in \eqref{MSE1} is not considered in the above optimization objective, since it is independent of the optimization variable $\mathbf{f}$. Upon invoking the various Karush-Kuhn-Tucker (KKT) conditions, the closed-form expression of the optimal TPC vector $\mathbf{f}^{\mathrm{opt}}$ can be determined as \cite{boyd2004convex}
\begin{align}\label{Opt_Precoder}
\mathbf{f}^{\mathrm{opt}}=\mathbf{\Psi}^{-1}\mathbf{Z}^H\left(\mathbf{Z}\mathbf{\Psi}^{-1}\mathbf{Z}^H\right)^{-1}\mathbf{b}.
\end{align}
Subsequently, the optimal digital TPC vector for the $m$th sensor, denoted by $\mathbf{f}_m^{\mathrm{opt}}$ can be obtained by extracting the sub-vector corresponding to the elements $\left[(m-1)qN_T+1\right]$ to $\left[mqN_T\right]$ from $\mathbf{f}^{\mathrm{opt}}$. Finally, the optimal TPC matrix corresponding to the $m$th IoTNo is obtained as $\mathbf{F}_m^{\mathrm{opt}}=\mathrm{vec}^{-1}_q\left(\mathbf{f}^{\mathrm{opt}}_m\right)$. The next subsection extends the above design paradigm to include both total and individual IoTNo power constraints.
\subsection{Precoder design under power constraints}
The average transmit power of the $m$th IoTNo can be expressed as
\begin{align}\label{eq2}
&\mathbb{E}\left\{\left\lvert\left\lvert\mathbf{F}_{\text{RF},m} \mathbf{F}_{\text{BB},m}\mathbf{x}_m\right\rvert\right\rvert^2\right\}\nonumber \\&=\mathrm{Tr}\left[\mathbb{E}\left\{\left(\mathbf{F}_{\text{RF},m} \mathbf{F}_{\text{BB},m}\mathbf{x}_m\right)\left(\mathbf{F}_{\text{RF},m} \mathbf{F}_{\text{BB},m}\mathbf{x}_m\right)^H\right\}\right]\nonumber \\
&=\mathrm{Tr}\left[\mathbf{F}_{\text{RF},m} \mathbf{F}_{\text{BB},m}\left(\mathbf{A}_m\mathbf{R}_{\theta}\mathbf{A}^H_m+\mathbf{R}_m\right)\mathbf{F}^H_{\text{BB},m}\mathbf{F}^H_{\text{RF},m} \right].
\end{align}
Therefore, the expression for the total average transmit power of the IoTNe is given by
\begin{align}\label{eq4}
&\sum_{m=1}^M\mathbb{E}\left\{\left\lvert\left\lvert\mathbf{F}_{\text{RF},m} \mathbf{F}_{\text{BB},m}\mathbf{x}_m\right\rvert\right\rvert^2\right\}\nonumber \\
&=\sum_{m=1}^M\mathrm{Tr}\left[\mathbf{F}_{\text{RF},m} \mathbf{F}_{\text{BB},m}\left(\mathbf{A}_m\mathbf{R}_{\theta}\mathbf{A}^H_m+\mathbf{R}_m\right)\mathbf{F}^H_{\text{BB},m}\mathbf{F}^H_{\text{RF},m} \right].
\end{align}
Exploiting now the identity in \eqref{Trace_vec}, the average transmit power expression of the $m$th IoTNo in \eqref{eq2} can be recast as
\begin{align}\label{eq7}
\mathbf{f}_m^H\left[\left(\mathbf{A}_m\mathbf{R}_{\theta}\mathbf{A}^H_m+\mathbf{R}_m\right)^T \otimes \mathbf{I}_{N_T} \right]\mathbf{f}_m=\mathbf{f}_m^H\mathbf{\Gamma}_m\mathbf{f}_m,
\end{align}
where we have $\mathbf{\Gamma}_m=\left[\left(\mathbf{A}_m\mathbf{R}_{\theta}\mathbf{A}^H_m+\mathbf{R}_m\right)^T \otimes \mathbf{I}_{N_T} \right] \in \mathbb{C}^{qN_T \times qN_T}$. Hence, the compact expression for the total power of the IoTNe can be formulated as 
\begin{align}\label{eq6}
\sum_{m=1}^{M}\mathbf{f}_m^H\mathbf{\Gamma}_m\mathbf{f}_m=\mathbf{f}^H\mathbf{\Gamma}\mathbf{f},
\end{align}
where we define the matrix $\mathbf{\Gamma}=\mathrm{diag}\left(\mathbf{\Gamma}_1,\mathbf{\Gamma}_2,\hdots,\mathbf{\Gamma}_M\right) \in \mathbb{C}^{MqN_T \times MqN_T}$. Therefore, the optimization problem that minimizes the resultant MSE at the FC constrained by the zero-forcing and total network power constraints in \eqref{eq5} and \eqref{eq6}, respectively, is given by
 \begin{eqnarray}\label{Total}
 \underset{\mathbf{f}}{\text{minimize}}& \mathbf{f}^H\mathbf{\Psi}\mathbf{f} \nonumber\\
 \text{subject to} & \mathbf{Z}\mathbf{f} = {\mathbf{b}}, \nonumber\\
 & \mathbf{f}^H \mathbf{\Gamma}\mathbf{f} \leq \rho_T,
 \end{eqnarray}
where the quantity $\rho_T$ denotes the total IoTNe transmit power level, which can be utilized for transmitting observations to the FC. Once again, upon applying the KKT framework\cite{boyd2004convex}, the optimal TPC vector $\mathbf{f}^{\mathrm{opt}}$ is obtained as 
\begin{align}\label{Opt_Precoder_Totol}
\mathbf{f}^{\mathrm{opt}}=\left[\mathbf{\Psi}+\lambda^{\mathrm{opt}}\mathbf{\Gamma}\right]^{-1}\mathbf{Z}^H\left[\mathbf{Z}^H \left[\mathbf{\Psi}+\lambda^{\mathrm{opt}}\mathbf{\Gamma}\right]^{-1}\mathbf{Z}\right]^{-1}\mathbf{b},
\end{align}
where $\lambda^{\mathrm{opt}}$ denotes the optimal dual variable corresponding to the inequality constraint. One can follow the procedure given after Eq. \eqref{Opt_Precoder} to determine the optimal TPC matrix for every IoTNo in the network from the optimal TPC vector $\mathbf{f}^{{\mathrm{opt}}}$. The MSE of parameter estimation in \eqref{MSE1} can also be minimized by taking each sensor's transmit power level into account. The corresponding optimization paradigm can be formulated  as
 \begin{eqnarray}\label{individual}
\underset{\mathbf{f}}{\text{minimize}}& \mathbf{f}^H\mathbf{\Psi}\mathbf{f} \nonumber\\
 \text{subject to} & \mathbf{Z}\mathbf{f} = {\mathbf{b}}, \nonumber\\
   & \mathbf{f}_m^H \mathbf{\Gamma}_m\mathbf{f}_m \leq \rho_m,\; 1 \leq m \leq M,
 \end{eqnarray}
where the quantity $\rho_m$ denotes the transmit power constraint of the $m$th IoTNo. Upon exploiting the convex nature of the above optimization problem, standard convex optimization solvers, such as CVX \cite{cvx}, can be employed for finding the solution. The corresponding optimal TPC matrix $\mathbf{F}_m^{\mathrm{opt}}$ can be obtained, once again, using a procedure similar to the one described after Eq.\eqref{Opt_Precoder}. The next subsection describes the procedure of decomposing each TPC matrix $\mathbf{F}_m$ into the corresponding RF and baseband TPC matrices $\mathbf{F}_{\text{RF},m}$ and $\mathbf{F}_{\text{BB},m}$, respectively, for each IoTNo $m$ in the IoTNe.
\subsection{SOMP based RF and baseband precoder design}
The optimal TPC vector in \eqref{Opt_Precoder} can be rewritten as,
\begin{equation}
\mathbf{f}^{\mathrm{opt}}=\mathbf{{\Psi}}^{-1}\mathbf{\tilde{b}},
\end{equation}
where $\tilde{\mathbf{b}}=\mathbf{Z}^H\left(\mathbf{Z}\mathbf{\Psi}^{-1}\mathbf{Z}^H\right)^{-1}\mathbf{b} \in \mathbb{C}^{MqN_T \times 1}$. Furthermore, the MMSE digital TPC vector $\mathbf{f}^{\mathrm{opt}}$ can be rearranged as 
\begin{align}
    \mathbf{f}^{\mathrm{opt}} = \begin{bmatrix}
\mathbf{f}_1^{\mathrm{opt}}\\
\mathbf{f}_2^{\mathrm{opt}}\\
\vdots\\
\mathbf{f}_M^{\mathrm{opt}}
\end{bmatrix}=
\begin{bmatrix}
\mathbf{\Psi}_1^{-1} & \mathbf{0} & \cdots & \mathbf{0}\\
\mathbf{0} & \mathbf{\Psi}_2^{-1} &  \cdots & \mathbf{0}\\
\vdots & & \ddots & \vdots\\
\mathbf{0} & \mathbf{0} & \cdots & \mathbf{\Psi}_M^{-1}
\end{bmatrix} \tilde{\mathbf{b}}.
\end{align}
It can be readily deduced from the above equation that the MMSE digital TPC vector corresponding to the $m$th IoTNo can be derived as $\mathbf{f}^{\mathrm{opt}}_m=\textrm{vec}\left(\mathbf{F}_m^{\mathrm{opt}}\right)=\mathbf{{\Psi}}_m^{-1}\tilde{\mathbf{b}}_m$. The quantity $\tilde{\mathbf{b}}_m$ can be obtained by extracting the $[(m-1)qN_T]$th element to the $[mqN_T]$th elements of the column vector $\tilde{\mathbf{b}}$. Furthermore, the optimal TPC matrix for the $m$th IoTNo can be determined as $\mathbf{F}_m^{\mathrm{opt}}=\mathrm{vec}^{-1}_q\left(\mathbf{f}_m\right)$. Subsequently, the problem of decomposing $\mathbf{F}_m^{\mathrm{opt}}$ into its baseband and RF components $\mathbf{F}_{\text{BB},m}$ and $\mathbf{F}_{\text{RF},m}$, respectively, can be formulated as the following best linear approximation problem:
 \begin{equation}
\begin{aligned} \label{E34}
\underset{\mathbf{F}_{\text{RF},m},{\mathbf{F}_{\text{BB},m}}}{\text{minimize}}\quad &\left\Vert{\mathbf{F}_m^{\mathrm{opt}}-\mathbf{F}_{\text{RF},m}{\mathbf{F}}_{\text{BB},m}}\right\Vert_{F}^2 \\
\textrm{subject to}
\quad & \left\lvert\left[\mathbf{F}_{\text{RF},m}\right]_{i,j}\right\rvert=\frac{1}{\sqrt{N_T}}.
\end{aligned}
\end{equation}
Note that the constant magnitude constraint for the elements of the RF TPC $\mathbf{F}_{\text{RF},m}$ makes the above problem non-convex and intractable. One can overcome this problem by exploiting the following interesting relationship between the MMSE TPC matrix $\mathbf{F}_m^{\mathrm{opt}}$ and the corresponding transmit array response matrix $\mathbf{A}_{s,m}$.
\begin{lem}
The column space of the fully digital MMSE TPC matrix $\mathbf{F}_m$ lies in the row space of the mmWave MIMO channel matrix $\mathbf{H}_m$, i.e.,
\begin{align}
\mathcal{C}(\mathbf{F}_m)\subseteq \mathcal{R}(\mathbf{H}_m),
\end{align}
where the row and column spaces of the matrices $\mathbf{H}_m$ and $\mathbf{F}_m$, corresponding to the $m$th sensor, are denoted by $\mathcal{R}(\mathbf{H}_m)$ and $\mathcal{C}(\mathbf{F}_m)$, respectively.
\end{lem}
\begin{proof}
The detailed proof is given in Appendix-\ref{Lemma1}. 
\end{proof}
Thus, one can choose the columns of $\mathbf{F}_{\text{RF},m}$ as the columns of the transmit array response matrix $\mathbf{A}_{s,m}$. The resultant optimization problem can be formulated as
\begin{equation}\label{EqSOMP}
\begin{aligned} 
\underset{{\widetilde{\mathbf{F}}_{\text{BB},m}}}{\text{minimize}} \quad & \left\Vert{\mathbf{F}_m^{\mathrm{opt}}-\mathbf{A}_{s,m}\widetilde{\mathbf{F}}_{\text{BB},m}}\right\Vert_{F}^2 \\
\textrm{subject to} \quad & 
\left\Vert \mathrm{diag}( \widetilde{\mathbf{F}}_{\text{BB},m}\widetilde{\mathbf{F}}_{\text{BB},m}^H) \right\Vert_0=N^m_{\text{RF}}, 
\end{aligned}
\end{equation}
where the constraint $\vert\vert \mathrm{diag}( \widetilde{\mathbf{F}}_{\text{BB},m}\widetilde{\mathbf{F}}_{\text{BB},m}^H) \vert\vert_0=N^m_{\text{RF}}$ results from the fact that only $N^m_{\text{RF}}$ out of $N_T$ rows must be non-zero, since there are only $N^m_{\text{RF}}$ RF chains at the $m$th IoTNo in the hybrid mmWave MIMO IoTNe. Thus, the matrix $\widetilde{\mathbf{F}}_{\text{BB},m} \in \mathbb{C}^{N_{cl} \times q}$ is block sparse in nature. The sparse signal recovery TPC design problem in \eqref{EqSOMP} can be readily solved by employing the popular SOMP algorithm. The algorithm terminates after a finite number of iterations, i.e. it requires $N^m_{\text{RF}}$ iterations to choose the $N^m_{\text{RF}}$ dominant transmit array response vectors $\mathbf{a}_{\text{s}}(\theta_{k,m})$ from the matrix $\mathbf{A}_{s,m}$. A brief summary of the algorithm follows. \par
The algorithm begins with step (5) that evaluates the projection of each column of $\mathbf{A}_{s,m}$ on every column of the residual matrix $\mathbf{F}_{\text{res}}$. Step (6) selects the specific column of $\mathbf{A}_{s,m}$ that is maximally correlated with the columns of $\mathbf{F}_{\text{res}}$. In step (7), the selected dominant vector is appended to the RF TPC matrix $\mathbf{F}_{\text{RF,m}}$. Upon employing this updated $\mathbf{F}_{\text{RF,m}}$, the classic least squares procedure is exploited for calculating the updated baseband TPC $\mathbf{F}_{\text{BB,m}}$ in step (8). The residual matrix $\mathbf{F}_{\text{res}}$ is updated in step (9). The above steps are repeated $N^m_{\text{RF}}$ times, followed by step (11) that scales the baseband TPC $\mathbf{F}_{\text{BB},m}$ to meet the transmit power constraint, to obtain the hybrid TPCs $\mathbf{F}_{\text{RF},m}$ and $\mathbf{F}_{\text{BB},m}$, corresponding to IoTNo $m$.  Finally, the block diagonal  $MN_T\times MN_{\text{RF}}$ RF TPC $\mathbf{F}_{\text{RF}}$  and the $MN_{\text{RF}}\times Mq$ baseband TPC $\mathbf{F}_{\text{BB}}$ are obtained using steps (13) and (14), respectively.
\begin{algorithm}[ht]
\caption{ Hybrid TPC design procedure for the mmWave MIMO IoTNe}\label{alg:algorithm1}
\begin{algorithmic}[1]
\Require{$\left\lbrace \mathbf{F}_m^{\mathrm{opt}}\right\rbrace$ and $N^m_{\text{RF}}$ },\; $\forall$ $m$
\For{$1 \leq m \leq M$}
\State $\mathbf{F}_{\text{RF},m}=[\  ]$
\State $\mathbf{F}_{\text{res}}=\mathbf{F}_m$
\For{$c\leq N^m_{\text{RF}}$}
   \State$\boldsymbol{\Phi}=\mathbf{A}_{s,m}^H\mathbf{F}_{\text{res}}$
   \State $n={\textrm{arg max}}_{k=1, ..., N_{cl}}[\boldsymbol{\Phi}\boldsymbol{\Phi}^H]_{k,k}$ 
   \State $\mathbf{F}_{\text{RF},m}=[\mathbf{F}_{\text{RF},m}\mid \mathbf{A}_{s,m}^{(n)}]$
   \State $\mathbf{F}_{\text{BB},m}=(\mathbf{F}_{\text{RF},m}^H\mathbf{F}_{\text{RF},m})^{-1}\mathbf{F}_{\text{RF},m}^H\mathbf{F}_m$
   \State $\mathbf{F}_{\text{res}}=\frac{\mathbf{F}_m-\mathbf{F}_{\text{RF},m}\mathbf{F}_{\text{BB},m}}{\|\mathbf{F}_m-\mathbf{F}_{\text{RF},m}\mathbf{F}_{\text{BB},m}\|_F}$
\EndFor
\State $\mathbf{F}_{{\text{BB}},m}=\frac{\mathbf{F}_{{\text{BB}},m}\left\lvert\left\lvert\mathbf{F}_{m}\right\rvert\right\rvert_F}{\left\lvert\left\lvert\mathbf{F}_{{\text{RF}},m}\mathbf{F}_{{\text{BB}},m}\right\rvert\right\rvert_F}$
\State \Return{$\mathbf{F}_{\text{RF},m}$ , $\mathbf{F}_{\text{BB},m}$}
\State $\mathbf{F}_{\text{RF}}=\text{diag}\left(\mathbf{F}_{\text{RF}},\mathbf{F}_{\text{RF},m}\right)$
\State $\mathbf{F}_{\text{BB}}=\text{diag}\left(\mathbf{F}_{\text{BB}},\mathbf{F}_{\text{BB},m}\right)$
\EndFor
\end{algorithmic}
\end{algorithm}
\subsection{Complexity Analysis}
The computational complexities of the various steps of Algorithm 1 are shown in Table \ref{complexity_algo1}.
The complexity of Algorithm 1 given the fully-digital transmit precoder matrix $\mathbf{F}_m$ is of the order of $\mathcal{O}\left(N_{\text{RF}}^mN_TN_{cl}q + (N_{\text{RF}}^m)^2N_T^2\right)$. As shown in Table \ref{complexity_table_algos}, the SOMP algorithm has a low complexity, i.e., $\mathcal{O}\bigg(\left(MqN_T\right)^3+M(N_{\text{RF}}^m)^2N_T^2\bigg)$  in contrast to the manifold optimization [14] or ADMM-based [30] iterative design frameworks, which have a complexity order of $\mathcal{O}\left(\left(MqN_T\right)^3+MN_{\text{RF}}^mN_T^3q^2N_{\text{inn}}N_{\text{out}}\right)$ and $\mathcal{O}\left(M(N_{\text{RF}}^m)^2q^4N_{\text{I}}+M(N_{\text{RF}}^m)^3N_T^3N_{\text{I}}\right)$, respectively.\\ 
\begin{minipage}[c]{0.5\textwidth}
\vspace{10pt}
\centering
\captionof{table}{Computational Complexity of Algorithm 1}
\label{complexity_algo1}
\begin{tabular}{|l|r|}
\hline
Expression & Computational complexity\\
\hline
$\boldsymbol{\Phi}$ & $\mathcal{O}\left(N_{\text{RF}}^mN_{cl}N_Tq\right)$ \\ \hline
 $\boldsymbol{\Phi}\boldsymbol{\Phi}^H$ & $\mathcal{O}\left(N_{\text{RF}}^mN_{cl}^2q\right)$ \\
\hline
$\mathbf{F}_{\text{RF},m}$ & $\mathcal{O}\left(N_{\text{RF}}^m N_{cl}N_T q + N_{\text{RF}}^m N_{cl}^2q\right)$ \\ \hline
$\mathbf{F}_{\text{BB},m}$ & $\mathcal{O}\left(\left(N_{\text{RF}}^m\right)^2N_T^2 + \left(N_{\text{RF}}^m\right)^2N_Tq\right)$ \\ \hline
$\mathbf{F}_{\text{res}}$ & $\mathcal{O}\left(N_{\text{RF}}^mN_Tq \right)$\\ \hline
\end{tabular}
\end{minipage}
\begin{minipage}[c]{0.5\textwidth}
\vspace{10pt}
\centering
\captionof{table}{Computational Complexity Comparison}
\label{complexity_table_algos}
\begin{tabular}{|l|r|}
\hline
Algorithm & Computational complexity\\
\hline
Proposed & $\mathcal{O}\bigg(\left(MqN_T\right)^3+M(N_{\text{RF}}^m)^2N_T^2\bigg)$ \\ \hline
MO-AltMin \cite{yu2016alternating} & $\mathcal{O}\left(\left(MqN_T\right)^3+MN_{\text{RF}}^mN_T^3q^2N_{\text{inn}}N_{\text{out}}\right)$\\
\hline
ADMM \cite{liu2021hybrid} & $\mathcal{O}\left(M(N_{\text{RF}}^m)^2q^4N_{\text{I}}+M(N_{\text{RF}}^m)^3N_T^3N_{\text{I}}\right)$\\ \hline
\end{tabular}
\end{minipage}
\subsection{Centralized MMSE Bound}
The centralized MMSE bound can be derived as described next in order to benchmark the performance of the proposed hybrid
beamforming designs. This is based on the principle that the best MSE performance is obtained, when the observations taken by every
IoTNo in the IoTNe are available directly at the FC without any degradation. For such a hypothetical system, the MMSE estimator  at the FC results in
the best MSE performance  of parameter estimation. The observations
x for the centralized scenario can be expressed as 
\begin{align}
\mathbf{x}=\mathbf{A}\boldsymbol{\theta}+\mathbf{v},
\end{align}
where the observation vector $\mathbf{x}  \in \mathbb{C}^{Mq \times 1}$ and noise vector $\mathbf{v} \in \mathbb{C}^{Mq \times 1}$ are defined as 
\begin{align}
\mathbf{x}&= [\mathbf{x}_1^H,\mathbf{x}_2^H,\hdots,\mathbf{x}_M^H]^H,\\
\mathbf{v}&= [\mathbf{v}_1^H,\mathbf{v}_2^H,\hdots,\mathbf{v}_M^H]^H.
\end{align}
Note that the noise vector $\mathbf{v}$ is distributed as $\mathbf{v} \sim \mathcal{CN} \left(\mathbf{0},\mathbf{R}_{v}\right)$, where the covariance matrix $\mathbf{R}_{v}\in \mathbb{C}^{Mq \times Mq}$ and the overall observation matrix $\mathbf{A} \in \mathbb{C}^{Mq \times p}$ are defined as 
\begin{align}
\mathbf{R}_{v}&=\mathrm{diag}\left[\mathbf{R}_{1},\mathbf{R}_{2},\hdots,\mathbf{R}_{M}\right], \\
\mathbf{A}&=[\mathbf{A}_1,\mathbf{A}_2,\hdots,\mathbf{A}_M].
\end{align}
The MSE of this centralized MMSE estimate is given as \cite{10.5555/151045}
\begin{equation}\label{MMSE_Bench}
\mathrm{MSE}_{\mathrm{MMSE}}=\mathrm{Tr}\left[\left(\mathbf{R}^{-1}_{\theta}+\mathbf{A}^H\mathbf{R}^{-1}_v\mathbf{A}\right)^{-1}\right],
\end{equation}
which is the centralized MMSE bound for the system under consideration. The next section develops robust hybrid beamformer design techniques in the presence of imperfect CSI. 
\section{Robust Hybrid Precoder Designs Relying on Imperfect CSI}\label{Imperfect_CSI}
Due to various factors such as a limited pilot overhead, quantization error etc., uncertainty in the available CSI is inevitable in practice. Therefore, for the practical viability of the designs developed, it is important to conceive robust hybrid pre-processing schemes that take the CSI uncertainty into account. The robust designs thus developed lead to an improved performance in comparison to their CSI uncertainty-agnostic counterparts. Since the uncertainty-agnostic design procedure does not account for the CSI uncertainty,  and designs the TPC/ RCs utilizing purely the available channel estimate, it results in an unacceptable performance degradation. The pertinent models characterizing the CSI error and the corresponding robust design procedures are described next. 
\subsection{Robust Precoder Design under Stochastic CSI Uncertainty}\label{Robust_Stochastic}
Employing the stochastic/ Gaussian CSI uncertainty model\cite{8906245,5356168,4567684}, the channel matrix $\mathbf{H}_m$ can be modelled as 
\begin{equation}\label{Channel_model}
    \mathbf{H}_m = \widehat{\mathbf{H}}_m+\underbrace{\mathbf{R}^{1/2}_{\text{FC}}\mathbf{S}_m\mathbf{R}^{T/2}_{s,m}}_{\Delta{\mathbf{H}}_m},
\end{equation}
where the matrix $\widehat{\mathbf{H}}_m \in \mathbb{C}^{N_R \times N_T}$ denotes the available channel estimate, whereas $\Delta{\mathbf{H}}_m \in \mathbb{C}^{N_R \times N_T}$ represents the estimation error. The matrices  $\mathbf{R}_{\text{FC}} \in \mathbb{C}^{N_R \times N_R}$ and $\mathbf{R}_{s,m} \in \mathbb{C}^{N_T \times N_T}$, represent the receive and transmit correlation matrices corresponding to the FC and the $m$th sensor, repectively\cite{5356168,4567684}. These can be modelled as $\mathbf{R}_{\text{FC}}(i,j)=\sigma_H^2 \alpha^{|i-j|}$ and $\mathbf{R}_{s,m}(i,j)=\beta^{|i-j|}_m$, where $\alpha$ and $\beta_m$ represent the receive and transmit correlation, and $\sigma_H^2$ denotes the uncertainty variance of the CSI. Each element of the matrix $\mathbf{S}_m \in \mathbb{C}^{N_R \times N_T}$  is assumed to be independent and identically (i.i.d.) distributed with mean $0$ and variance one. Substituting $\mathbf{H}_m$ according to the channel model in \eqref{Channel_model}, the equivalent received vector $\widetilde{\mathbf{y}}$ obtained after RF combining at the FC can be rewritten as
\begin{align}
\widetilde{\mathbf{y}}&=\sum_{m = 1}^{M} \mathbf{W}_{\text{RF}}^H \left(\widehat{\mathbf{H}}_m+\Delta{\mathbf{H}}_m\right) \mathbf{F}_m \mathbf{A}_m\boldsymbol{\theta}  \nonumber \\&+ \sum_{m = 1}^{M} \mathbf{W}_{\text{RF}}^H \left(\widehat{\mathbf{H}}_m+\Delta{\mathbf{H}}_m\right) \mathbf{F}_m \mathbf{v}_m +{\mathbf{W}_{\text{RF}}^H\mathbf{u}}.
 \end{align}
The estimation constraint under imperfect CSI availability can be reformulated as 
 \begin{align}\label{Mod_Est_Cons}
\sum_{m=1}^{M}\mathbf{W}_{\text{RF}}^H \widehat{\mathbf{H}}_m \mathbf{F}_m  \mathbf{A}_m=\mathbf{I}_{p}.
\end{align}
Applying the $\mathrm{vec}$ operation at both sides and using the identity $\mathrm{vec}\left(\mathbf{N}\mathbf{P}\mathbf{Q}\right)=\left(\mathbf{Q}^T \otimes \mathbf{N}\right)\mathrm{vec}(\mathbf{P})$, the above constraint can be rewritten as
\begin{align}\label{Mod_Est_Cons_n}
\sum_{m = 1}^{M} \left[\mathbf{A}_m^T \otimes \mathbf{W}_{\text{RF}}^H  \widehat{\mathbf{H}}_m \right] \mathrm{vec}\left[\mathbf {F}_m\right]=\sum_{m = 1}^{M} \mathbf{W}_m \mathbf{f}_m=\mathbf{W}\mathbf{f}=\mathbf{c},
\end{align}
where we have $\mathbf{W}_m=\left[\mathbf{A}_m^T \otimes \mathbf{W}_{\text{RF}}^H  \widehat{\mathbf{H}}_m \right] \in \mathbb{C}^{p^2 \times qN_T}$ and $\mathbf{W}=[\mathbf{W}_1, \mathbf{W}_2,\hdots,\mathbf{W}_M] \in \mathbb{C}^{p^2 \times MqN_T}$. The resultant MSE of parameter estimation is given by
\begin{align}\label{MSE_2}
\mathrm{MSE}&=\sum_{m = 1}^{M} \mathrm{Tr} \Big[\mathbf{W}_{\text{RF}}^H \Delta\mathbf{H}_m \mathbf{F}_m \mathbf{A}_m \mathbf{R}_{\boldsymbol{\theta}} \mathbf{A}_m^H \mathbf{F}_m^H  \Delta\mathbf{H}_m^H \mathbf{W}_{\text{RF}} \Big]\nonumber \\ &+ \mathrm{Tr} \Big[\mathbf{W}_{\text{RF}}^H  \mathbf{H}_m \mathbf{F}_m \mathbf{R}_{m} \mathbf{F}_m^H  \mathbf{H}_m^H \mathbf{W}_{\text{RF}} \Big]+\mathrm{Tr}\left[\mathbf{W}_{\text{RF}}^H \mathbf{R}_u \mathbf{W}_{\text{RF}}\right].
\end{align}
The following lemma can now be used to derive the average-MSE defined as $\overline{\mathrm{MSE}}=\mathbb{E}\left[\mathrm{MSE}\right]$:
\begin{lem}
For a matrix $\mathbf{X} \in \mathbb{C}^{r \times t}$, which obey the distribution $\mathbf{X} \sim \mathcal{CN}\left(\widehat{\mathbf{X}},\mathbf{R}_t \otimes \mathbf{R}_r\right)$, where $\mathbf{R}_r \in \mathbb{C}^{r \times r}$ and $\mathbf{R}_t \in \mathbb{C}^{t \times t}$ represent the receive and transmit correlation matrices, respectively, and a compatible matrix $\mathbf{Z}$, it follows that\cite{Gupta2018matrix}
\begin{align}
\mathbb{E}\left[\mathbf{X}\mathbf{Z}\mathbf{X}^H\right]&=\widehat{\mathbf{X}}\mathbf{Z}\widehat{\mathbf{X}}^H+\mathrm{Tr}\left[\mathbf{Z}\mathbf{R}_r\right]\mathbf{R}^T_t \nonumber \\
\mathbb{E}\left[\mathbf{X}^H\mathbf{Z}\mathbf{X}\right]&=\widehat{\mathbf{X}}^H\mathbf{Z}\widehat{\mathbf{X}}+\mathrm{Tr}\left[\mathbf{Z}\mathbf{R}^T_t\right]\mathbf{R}_r.
\end{align}
\end{lem}
\hspace{-9pt}Employing the above lemma, one can simplify the MSE expression in \eqref{MSE_2} to determine the average MSE as 
\begin{align}
&\overline{\mathrm{MSE}}=\sum_{m = 1}^{M} \mathrm{Tr}\left[\mathbf{F}_m \mathbf{A}_m \mathbf{R}_{\boldsymbol{\theta}} \mathbf{A}_m^H \mathbf{F}_m^H \mathbf{R}^T_{s,m}\right]\mathrm{Tr}\left[\mathbf{W}_{\text{RF}}^H\mathbf{R}^T_{\text{FC}}\mathbf{W}_{\text{RF}}\right] \nonumber \\&+\mathrm{Tr}\left[\mathbf{W}_{\text{RF}}^H  \widehat{\mathbf{H}}_m \mathbf{F}_m \mathbf{R}_{m} \mathbf{F}_m^H  \widehat{\mathbf{H}}_m^H \mathbf{W}_{\text{RF}}\right] +\mathrm{Tr}\left[ \mathbf{F}_m \mathbf{R}_{m} \mathbf{F}_m^H \mathbf{R}^T_{s,m}\right]\nonumber \\ &\left[\mathbf{W}_{\text{RF}}^H\mathbf{R}^T_{\text{FC}}\mathbf{W}_{\text{RF}}\right]+ \mathrm{Tr}\left[\mathbf{W}_{\text{RF}}^H\mathbf{R}_{u}\mathbf{W}_{\text{RF}}\right].
\end{align}
By exploiting the property in \eqref{Trace_vec}, the above expression can be further simplified to:
\begin{align}
\overline{\mathrm{MSE}}&=\sum_{m = 1}^{M} \mathbf{\alpha} \mathbf{f}_m^H \left[\left(\mathbf{A}_m \mathbf{R}_{\theta} \mathbf{A}_m^H \right)^T \otimes \mathbf{R}_{s,m}^T\right]\mathbf{f}_m \nonumber \\&+\mathbf{f}_m^H \bigg[\mathbf{R}_{m}^T \otimes \widehat{\mathbf{H}}_m^H \mathbf{W}_{\text{RF}} \mathbf{W}_{\text{RF}}^H\widehat{\mathbf{H}}_m\bigg]\mathbf{f}_m \nonumber \\& +\alpha \mathbf{f}_m^H \left[\mathbf{R}_{m}^T \otimes \mathbf{R}_{s,m}^T \right]\mathbf{f}_m+\mathrm{Tr}\left[\mathbf{W}_{\text{RF}}^H \mathbf{R}_u \mathbf{W}_{\text{RF}}\right]. 
\end{align}
The above expression can be recast as
\begin{align}\label{Avg_MSE}
\overline{\mathrm{MSE}}&=\sum_{m = 1}^{M} \mathbf{f}^H_m \mathbf{\Omega}_m \mathbf{f}_m +\mathrm{Tr}\left[\mathbf{W}_{\text{RF}}^H \mathbf{R}_u \mathbf{W}_{\text{RF}}\right]\nonumber \\ &= \mathbf{f}^H \mathbf{\Omega} \mathbf{f} +\mathrm{Tr}\left[\mathbf{W}_{\text{RF}}^H \mathbf{R}_u \mathbf{W}_{\text{RF}}\right],
\end{align}
wherein $\mathbf{J}_m=\left(\mathbf{A}_m \mathbf{R}_{\theta} \mathbf{A}_m^H \right)^T \otimes \mathbf{R}_{s,m}^T  \in \mathbb{C}^{qN_T \times qN_T } $, $\mathbf{T}_m=\mathbf{R}_{m}^T \otimes \mathbf{R}_{s,m}^T  \in \mathbb{C}^{qN_T  \times qN_T }$, $\mathbf{L}_m=\mathbf{R}_{m}^T \otimes \widehat{\mathbf{H}}_m^H \mathbf{W}_{\text{RF}} \mathbf{W}_{\text{RF}}^H\widehat{\mathbf{H}}_m \in \mathbb{C}^{qN_T  \times qN_T }$, $\mathbf{\Omega}_m=\mathbf{L}_m+\mathbf{\alpha}\left(\mathbf{J}_m+\mathbf{T}_m\right) \in \mathbb{C}^{qN_T  \times qN_T }$ and the scalar quantity $\alpha=\mathrm{Tr}\left[\mathbf{W}_{\text{RF}}^H\mathbf{R}^T_{\text{FC}}\mathbf{W}_{\text{RF}}\right]$. The block diagonal matrix $\mathbf{\Omega} \in \mathbb{C}^{MqN_T \times MqN_T}$ is defined as
\begin{align}
\mathbf{\Omega}=\mathrm{diag}\left[\mathbf{\Omega}_1,\mathbf{\Omega}_2,\hdots,\mathbf{\Omega}_M\right].
\end{align}
Hence, the optimization problem of minimizing the MSE in \eqref{Avg_MSE} subject to the modified estimation constraint in \eqref{Mod_Est_Cons_n} is given by
\begin{eqnarray}\label{eqaa'}
\underset{\mathbf{f}}{\text{minimize}}& \mathbf{f}^H\mathbf{\Omega}\mathbf{f} \nonumber\\
\text{subject to} & {\mathbf{W}}{\mathbf{f}} = {\mathbf{c}}.
\end{eqnarray}
Similar to \eqref{MSE_Min}, the above average MSE minimization problem can also be solved using the popular KKT framework. The corresponding closed-form solution of the robust TPC vector can be attained by replacing the matrices $\mathbf{\Psi}$, $\mathbf{Z}$ and the vector $\mathbf{b}$ in \eqref{Opt_Precoder} by the matrix $\mathbf{\Omega}$, $\mathbf{W}$ and $\mathbf{c}$, respectively. Subsequently, one can follow the procedure mentioned after \eqref{Opt_Precoder} to determine the TPC matrices $\mathbf{F}_m$ corresponding to each IoTNo $m$. Furthermore, the SOMP algorithm can once again be employed for decomposing the robust TPC matrix $\mathbf{F}_m$ to obtain the RF and baseband TPC matrices corresponding to each IoTNo $m$. The next subsection models the CSI uncertainty using the popular norm ball model\cite{1468482}, followed by the robust TPC design to mitigate the resultant performance loss.
\subsection{Robust Precoder Design under Norm Ball CSI Uncertainty}\label{Robust_Bounded}
The popular norm ball/ bounded CSI uncertainty model \cite{1468482} for the channel between each IoTNo $m$ and the FC is given by
\begin{equation}
\mathbf{H}_m= \widehat{\mathbf{H}}_m+{\Delta{\mathbf{H}}_m},
\end{equation}
where ${\| {\Delta{\mathbf{H}}}_m\|}_F \leq \epsilon_{H}$. The quantity $\epsilon_{H}$ represents the uncertainty radius. Interestingly, the norm ball CSI uncertainty model enables the minimization of the worst-case MSE, thus ensuring robustness. Toward this end, once again exploiting the property in (14), the MSE expression in \eqref{MSE_2} can be recast as
\begin{equation}\label{111}
\mathrm{MSE}=\sum_{m = 1}^{M} \mathbf{f}_m^H \mathbf{D}_m \mathbf{f}_m+\mathbf{f}_m^H \mathbf{S}_m \mathbf{f}_m +\mathrm{Tr}\left[\mathbf{W}_{\text{RF}}^H \mathbf{R}_u \mathbf{W}_{\text{RF}}\right],
\end{equation}
where the matrices  $\mathbf{S}_m \in\mathbb{C}^{qN_T \times qN_T}$ and $\mathbf{D}_m \in \mathbb{C}^{qN_T\times qN_T}$ are defined as
\begin{align}
\mathbf{D}_m&=\left[\left(\mathbf{A}_m\mathbf{R}_{\theta} \mathbf{A}_m^H \right)^T\otimes\Delta\mathbf{H}_m^H \mathbf{W}_{\text{RF}}\mathbf{W}_{\text{RF}}^H\Delta\mathbf{H}_m\right],\nonumber \\
\mathbf{S}_m&=\left[\mathbf{R}_{m}^T\otimes\mathbf{H}_m^H \mathbf{W}_{\text{RF}} \mathbf{W}_{\text{RF}}^H\mathbf{H}_m \right].
\end{align}
By exploiting the property $\left(\mathbf{AB} \otimes \mathbf{CD}\right)=\left(\mathbf{A} \otimes  \mathbf{C}\right)\left(\mathbf{B} \otimes  \mathbf{D}\right)$ \cite{zhang2017matrix}, the matrices $\mathbf{D}_m$ and $\mathbf{S}_m$ can be decomposed as
\begin{align}
\mathbf{D}_m&=\mathbf{G}_m^H \mathbf{G}_m, \\
\mathbf{S}_m&=\mathbf{L}_m^H \mathbf{L}_m,
\end{align}
where the matrices $\mathbf{G}_m=\left[\left( \mathbf{R}_{\theta}^{1/2}\mathbf{A}_m^H\right)\otimes\mathbf{W}_{\text{RF}}^H\Delta\mathbf{H}_m\right] \in\mathbb{C}^{qN_T \times qN_T}$ and $\mathbf{L}_m=\left[\mathbf{R}_{m}^{1/2} \otimes  \mathbf{W}_{\text{RF}}^H\mathbf{H}_m \right] \in \mathbb{C}^{qN_T\times qN_T}$. The MSE expression in \eqref{111} can be further simplified as
\begin{align}\label{11111}
\mathrm{MSE}&=\sum_{m = 1}^{M} \mathbf{f}_m^H \mathbf{G}_m^H \mathbf{G}_m \mathbf{f}_m+\mathbf{f}_m^H \mathbf{L}_m^H \mathbf{L}_m \mathbf{f}_m +\mathrm{Tr}\left[\mathbf{W}_{\text{RF}}^H \mathbf{R}_u \mathbf{W}_{\text{RF}}\right]\nonumber \\
&=\sum_{m = 1}^{M} {\| \mathbf{G}_m \mathbf{f}_m\|}^2+{\| \mathbf{L}_m \mathbf{f}_m\|}^2+\mathrm{Tr}\left[\mathbf{W}_{\text{RF}}^H \mathbf{R}_u \mathbf{W}_{\text{RF}}\right]\nonumber \\
&={\| \mathbf{G}\mathbf{f}\|}^2+{\| \mathbf{L}\mathbf{f}\|}^2+\mathrm{Tr}\left[\mathbf{W}_{\text{RF}}^H \mathbf{R}_u \mathbf{W}_{\text{RF}}\right],
\end{align}
where the block-diagonal matrices $\mathbf{G}\in \mathbb{C}^{MqN_T\times MqN_T}$ and $\mathbf{L}\in \mathbb{C}^{MqN_T\times MqN_T}$ are defined as $\mathbf{G}=\mathrm{diag}\left[\mathbf{G}_1,\mathbf{G}_2,\hdots,\mathbf{G}_M\right]$ and $\mathbf{L}=\mathrm{diag}\left[\mathbf{L}_1,\mathbf{L}_2,\hdots,\mathbf{L}_M\right]$, respectively. By exploiting the property $ \mathbf{\| Ax\|}_2 \leq \mathbf{\| A\|}_F \mathbf{\| x\|}_2$ \cite{zhang2017matrix}, the first term of the MSE expression in \eqref{11111} can be further bounded as 
\begin{equation}
{\| \mathbf{G}\mathbf{f}\|} \leq {\| \mathbf{G}\|}_F {\| \mathbf{f}\|}_2,
\end{equation}
where the quantity $\mathbf{\| G\|}_F=\left[\mathrm{Tr} \left[\mathbf{G}^H \mathbf{G}\right]\right]^{1/2}=\left[\sum_{m=1}^M\mathrm{Tr}\left[\mathbf{G}^H_m\mathbf{G}_m\right]\right]^{1/2}$. Furthermore, exploiting the property $\mathrm{Tr}\left[\mathbf{A} \otimes \mathbf{B}\right]=\mathrm{Tr}\left[\mathbf{A}\right]\mathrm{Tr}\left[\mathbf{B}\right]$, one can simplify the quantity $\mathbf{\| G\|}^2_F$ as
\begin{align}
&\mathbf{\| G\|}^2_F=\sum_{m = 1}^{M} \mathrm{Tr} \left[\left[\mathbf{A}_m \mathbf{R}_{\theta} \mathbf{A}_m^H \right]^T\right] \mathrm{Tr}\left[\Delta\mathbf{H}_m^H \mathbf{W}_{\text{RF}} \mathbf{W}_{\text{RF}}^H \Delta\mathbf{H}_m\right] \nonumber \\
&\leq  \epsilon_{H}^2 \left[\sum_{m = 1}^{M} \mathrm{Tr} \left[\mathbf{A}_m \mathbf{R}_{\theta} \mathbf{A}_m^H \right]^T \lambda_{max}\left(\mathbf{W}_{\text{RF}} \mathbf{W}_{\text{RF}}^H  \right)\right] = \eta^2,
\end{align}
where the above inequality follows from $\mathrm{Tr} \left[\mathbf{A} \mathbf{B}\right] \leq \lambda_{\mathrm{max}}\left(\mathbf{A}\right)\mathrm{Tr}\left[\mathbf{B}\right]$ and $\mathrm{Tr} \left[\Delta\mathbf{H}_m \Delta\mathbf{H}_m^H\right]\leq \epsilon_{H}^2$ \cite{zhang2017matrix}. Hence, one can upper-bound ${\| \mathbf{G}\mathbf{f}\|}^2_F$ as
\begin{equation}
 {\| \mathbf{G}\mathbf{f}\|}^2_F \leq \eta^2 {\| \mathbf{f}\|}_2^2.
 \end{equation}
 \begin{figure*}
\centering
\subfloat[]{\includegraphics[width=8.5cm, height=6.75 cm]{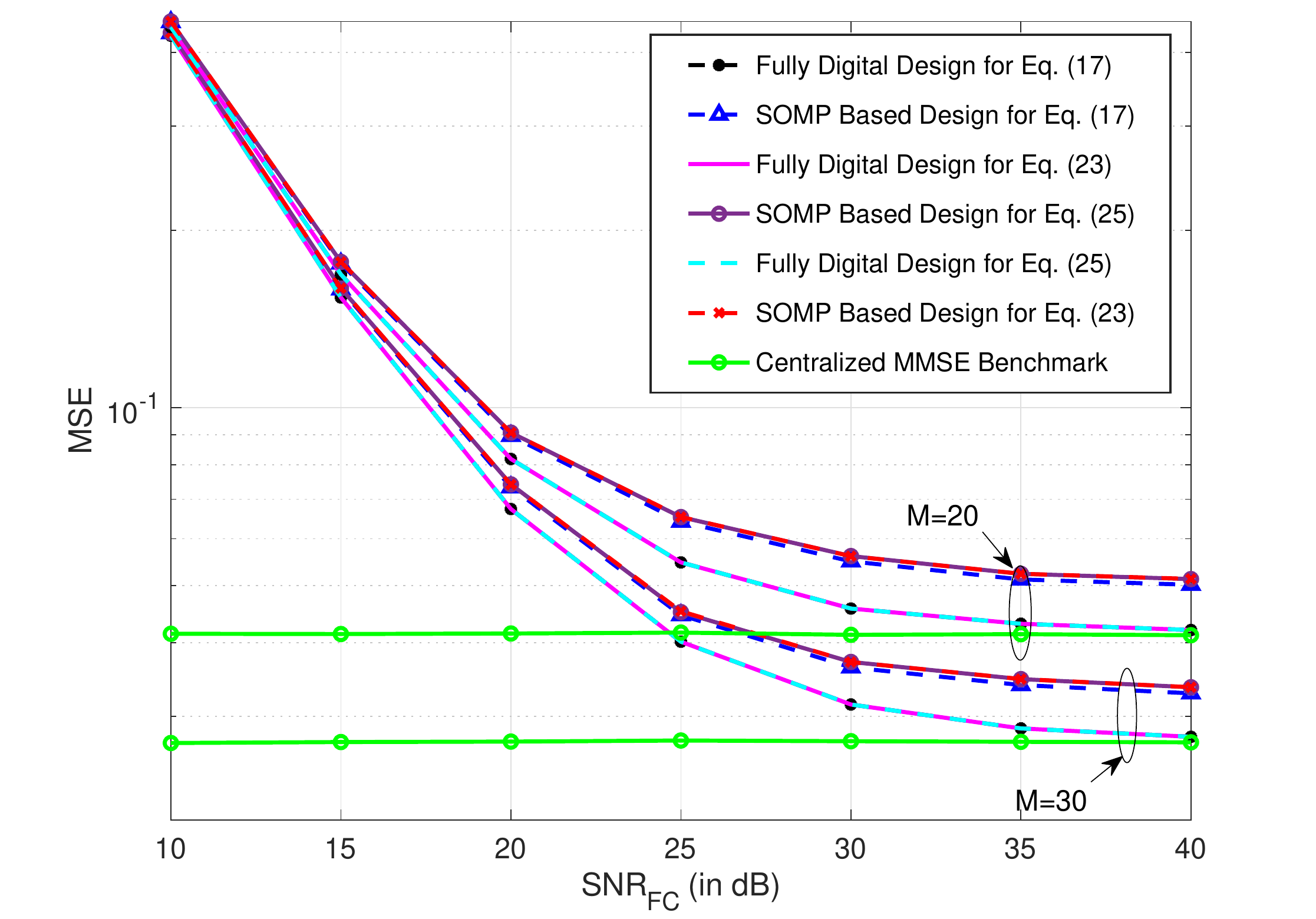}} 
\hfil
\hspace{-15pt}\subfloat[]{\includegraphics[width=8.5cm, height=6.75 cm]{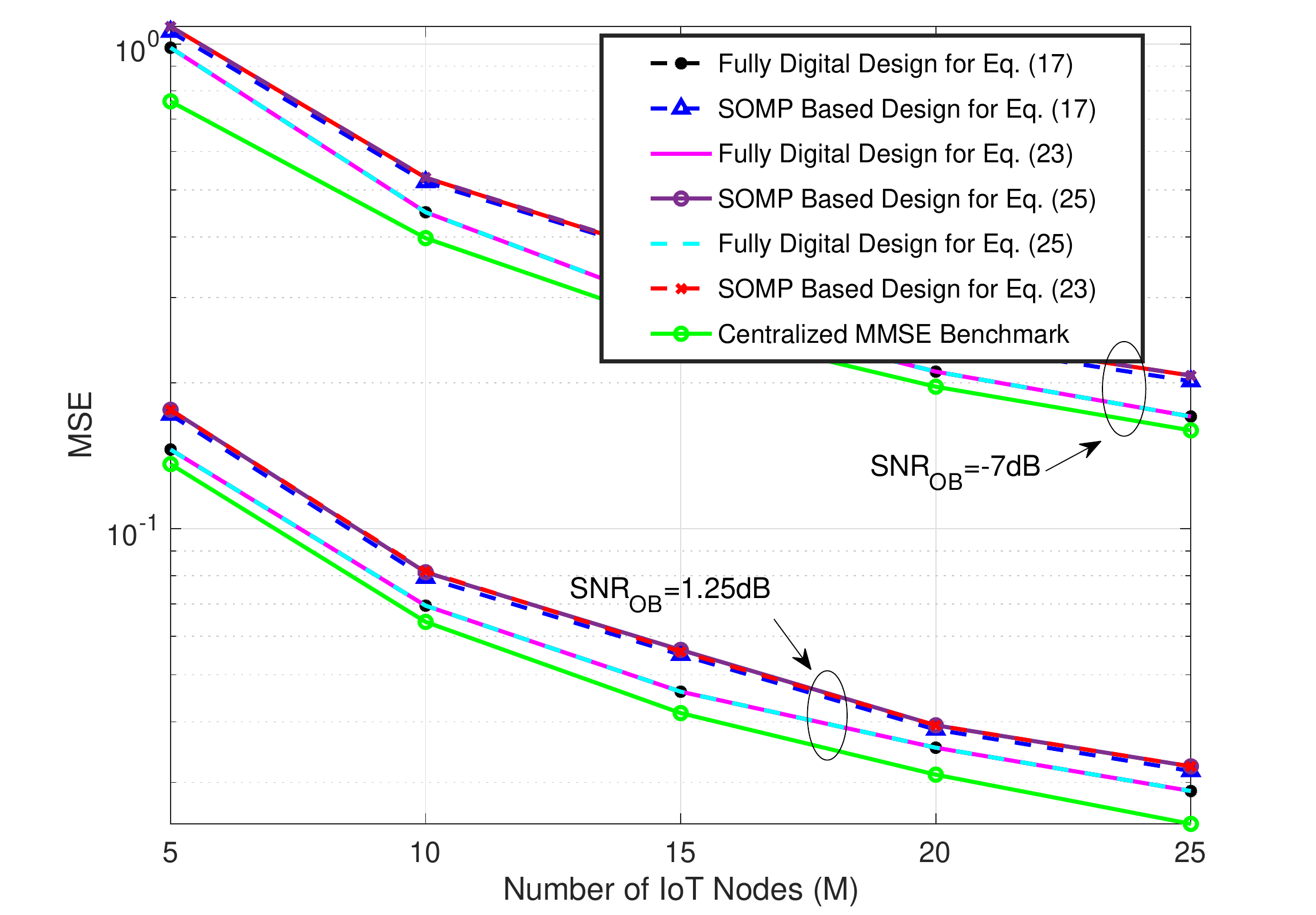}}
\hfil
\caption{$\left(a\right)$ MSE versus $\text{SNR}_{\text{FC}}$ {$ \left(b\right)$ MSE performance with varying number of IoTNos ($M$) for the scenario with perfect CSI.}}
\label{MVDP}
\end{figure*}
Substituting $\mathbf{H}_m = \widehat{\mathbf{H}}_m+\Delta\mathbf{H}_m$ into the expression of matrix $\mathbf{L}_m$, one obtains  $\mathbf{L}_m= \widehat{\mathbf{L}}_m+\Delta\mathbf{L}_m$, where 
\begin{align*}
\widehat{\mathbf{L}}_m=\left[\mathbf{R}_{m}^{1/2} \otimes  \mathbf{W}_{\text{RF}}^H \widehat{\mathbf{H}}_m \right],
\Delta{\mathbf{L}}_m=\left[\mathbf{R}_{m}^{1/2} \otimes  \mathbf{W}_{\text{RF}}^H \Delta{\mathbf{H}}_m\right].
\end{align*}
The term $ {\| \mathbf{L}\mathbf{f}\|}^2$ can be upper-bounded as 
 \begin{align}\label{1111}
 \mathbf{\| Lf\|}^2&= \left\lvert \left\lvert\left(\widehat{\mathbf{L}}+\Delta{\mathbf{L}}\right)\mathbf{f}\right\rvert \right\rvert^2 \leq \left(\mathbf{\left\lvert \left\lvert \widehat{\mathbf{L}}f\right\rvert \right\rvert}+\mathbf{\left\lvert \left\lvert \Delta\mathbf{L} f\right\rvert \right\rvert}\right)^2 \nonumber \\&\leq \left(\mathbf{\left\lvert \left\lvert \widehat{\mathbf{L}}f\right\rvert \right\rvert}+\mathbf{\left\lvert \left\lvert\Delta{L}\right\rvert \right\rvert}_F \mathbf{\| f\|} \right)^2,
 \end{align}
where the first inequality follows from the triangle inequality and the second inequality follows from $ \mathbf{\| Ax\|}_2 \leq \mathbf{\| A\|}_F \mathbf{\| x\|}_2$. The quantity $\mathbf{\| \Delta{L}\|}_F^2$ can be further upper bounded as
 \begin{align}
\mathbf{\| \Delta{L}\|}_F^2&=\mathrm{Tr} \left[\Delta\mathbf{L} \Delta\mathbf{L}^H\right]\nonumber \\ &={\sum_{m = 1}^{M} \mathrm{Tr} \left[ \mathbf{R}_{m}^T \otimes  \Delta\mathbf{H}_m^H \mathbf{W}_{\text{RF}} \mathbf{W}_{\text{RF}} \Delta\mathbf{H}_m \right]}\leq \zeta^2,
\end{align} 
where the constant $\zeta$ is defined as
\begin{align}
\zeta=\epsilon_{H} \left[\sum_{m = 1}^{M} \mathrm{Tr} \left[ \mathbf{R}_{m}\right] \lambda_{\text{max}}\left[\mathbf{W}_{\text{RF}} \mathbf{W}_{\text{RF}}^H  \right]\right]^{1/2}.
\end{align}
Hence, the upper bound in \eqref{1111} reduces to
\begin{align}
\mathbf{\| Lf\|}^2 \leq \left[\mathbf{\left\lvert \left\lvert \widehat{\mathbf{L}}f\right\rvert \right\rvert}+\zeta \mathbf{\| f\|} \right]^2.
\end{align}
Finally, the MSE in \eqref{11111} can be upper-bounded as
\begin{equation}\label{wors_MSE}
\mathrm{MSE} \leq \left[\mathbf{\left\lvert \left\lvert \widehat{\mathbf{L}}f\right\rvert \right\rvert}+\zeta \mathbf{\| f\|} \right]^2 + \eta^2 \mathbf{\| f\|}_2.
 \end{equation}
Using the deductions above, the optimization problem of minimizing the MSE in \eqref{wors_MSE} obeying the modified constraint in \eqref{Mod_Est_Cons} is formulated as 
\begin{eqnarray}\label{eqaa'}
\underset{\mathbf{f}}{\text{minimize}}& \left[\mathbf{\left\lvert \left\lvert \widehat{\mathbf{L}}f\right\rvert \right\rvert}+\zeta \mathbf{\| f\|} \right]^2 + \eta^2 \mathbf{\| f\|}_2 \nonumber\\
\text{subject to} & {\mathbf{W}}{\mathbf{f}} = {\mathbf{c}}.
\end{eqnarray}
The convexity of the above optimization problem renders its solution easy to determine using a suitable convex solver. Once again, one can follow the procedure mentioned after \eqref{Opt_Precoder} to determine the TPC matrices $\mathbf{F}_m$ corresponding to each IoTNo $m$. Furthermore, the SOMP algorithm can once again be employed for factorizing the robust TPC matrix $\mathbf{F}_m$ to obtain the RF and baseband TPC matrices corresponding to each IoTNo $m$. The overall procedure of linear decentralized estimation in an IoTNe can be summarized as follows.
Since, the IoTNos are small, battery-powered devices with a finite battery life and limited processing power. The FC, by contrast, does not face such power constraints and typically has higher processing and communication capacities to meet the demands of the IoT setup. The precise steps for obtaining the instantaneous CSI and parameter estimates are as follows. 
\begin{itemize}
    \item Each IoTNo begins by transmitting pilot symbols to the FC during each coherence time period for channel estimation.
    \item On reception of the pilots from each IoTNo $m$, the FC estimates the channels $\lbrace\mathbf{H}_m\rbrace_{m=1}^M$ corresponding to each IoTNo. 
    \item  Following this step, the FC designs the hybrid TPC matrices for each IoTNo and feeds back to each IoTNo its hybrid TPC.
\end{itemize}
In addition, it is also important to note that the FC is only required to have access only to the statistical CSI, such as the observation noise covariance matrix $\mathbf{R}_{m}$ and the vector parameter covariance matrix $\mathbf{R}_{\theta}$, which may be obtained by averaging over an appropriately long interval of time. Additionally, since a quasi-static mmWave MIMO channel is considered in this work, the CSI obtained at the FC remains constant over a period of time. Hence, the hybrid TPCs $\mathbf{F}_{\text{BB},m}$, $\mathbf{F}_{\text{RF},m}$ does not have to be computed and fed-back frequently. To further reduce the overhead of CSI feedback, robust hybrid TPC designs have been proposed that use either imperfect or a rough estimate of the CSI. The robust model captures the relationship between the overhead in the number of pilot symbols required and the channel estimation error. For example, in the stochastic CSI uncertainty model,  the channel estimation error decreases upon increasing the number of pilot symbols. The next section discusses the numerical results characterizing the performance of the various designs proposed in this work.
\begin{figure}
\centering
\includegraphics[width=8.5cm, height=6.75 cm]{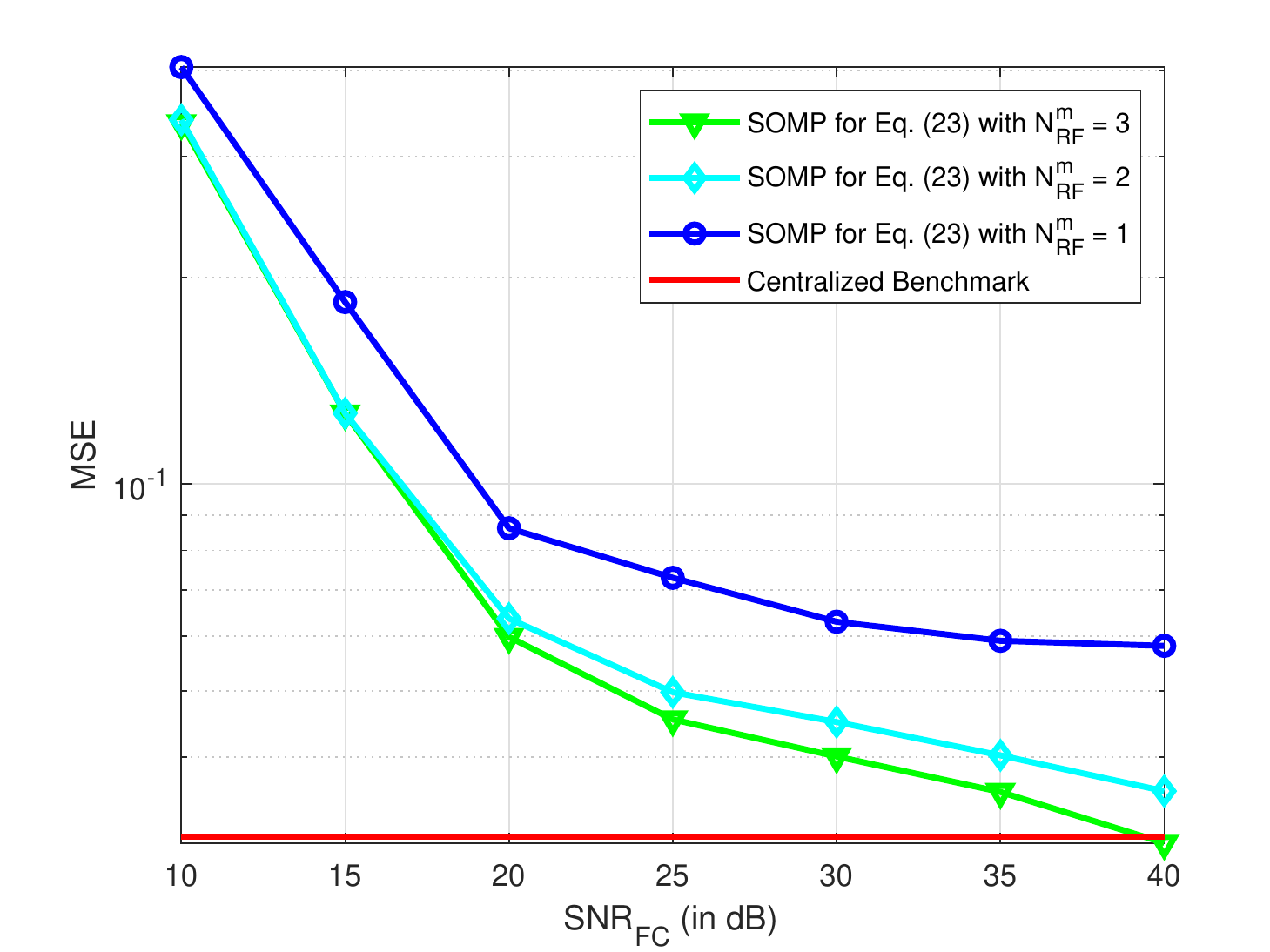}
\caption{MSE versus $\text{SNR}_{\text{FC}}$ for different values of RF chains $N_{\text{RF}}^m$ at each IoTNo.}
\label{fig:mse_snr_rf}
\end{figure}
\begin{figure*}
\centering
\subfloat[]{\includegraphics[width=8.5cm, height=6.75 cm]{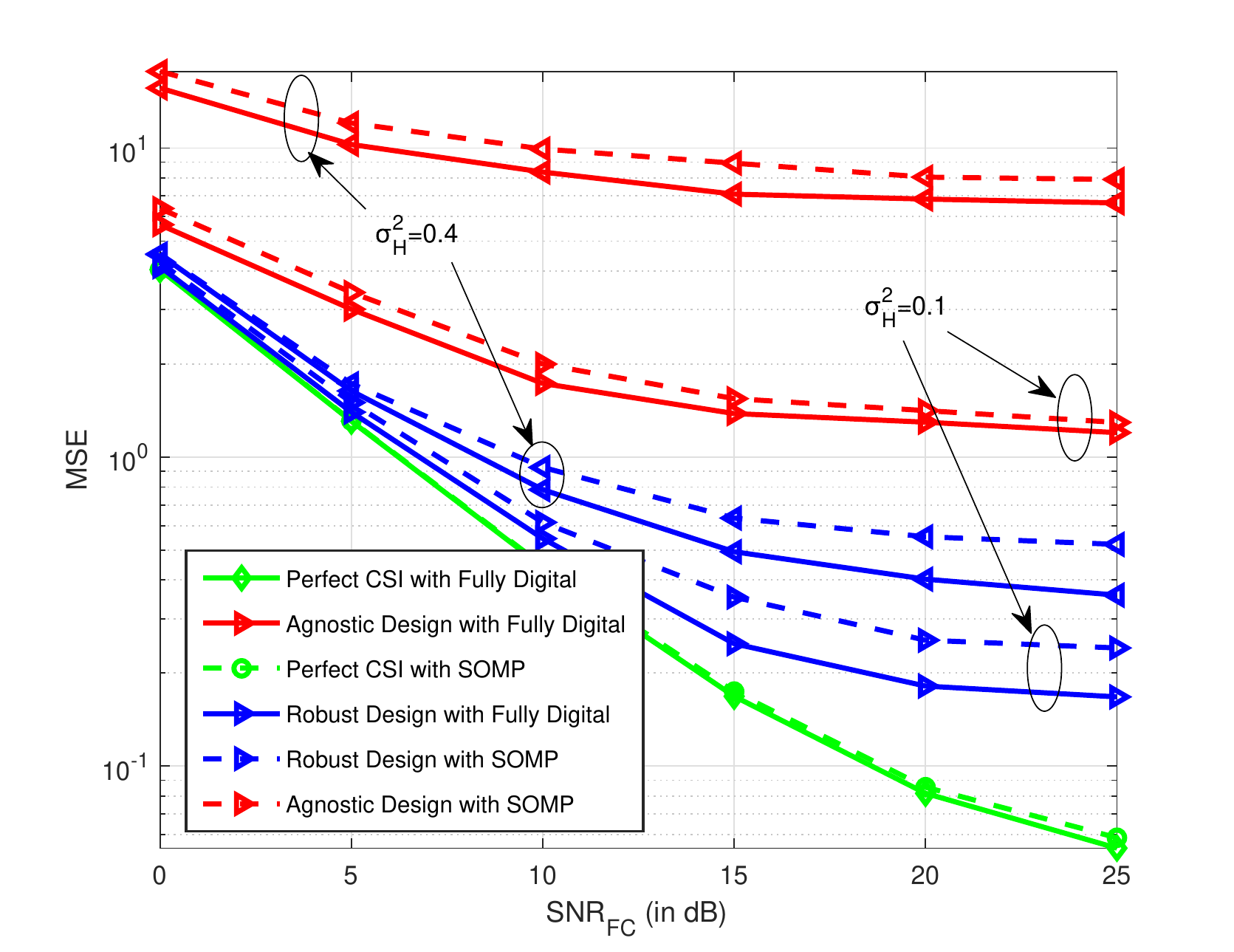}}
\hfil
\hspace{-15pt}\subfloat[]{\includegraphics[width=8.5cm, height=6.75 cm]{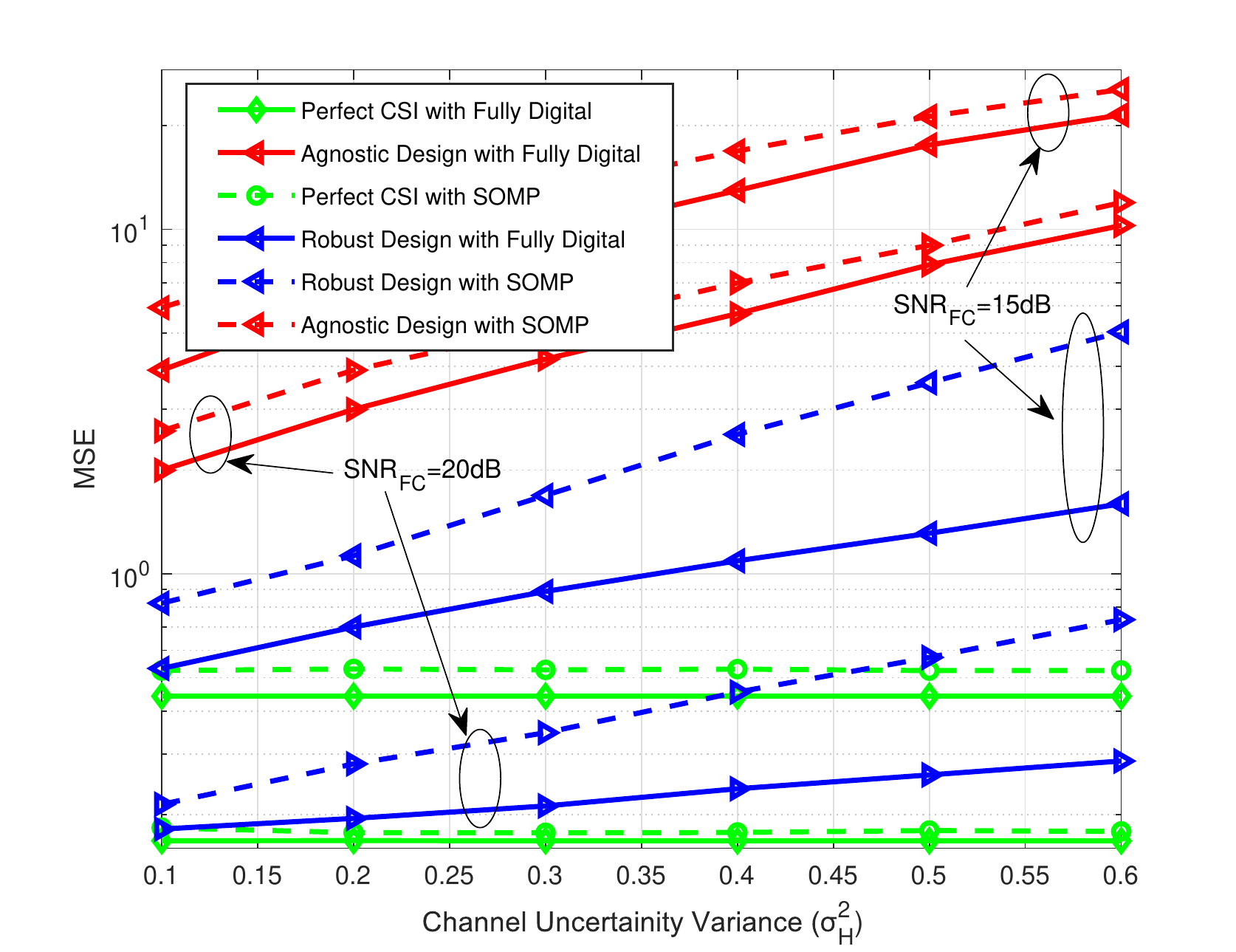}}
\hfil
\caption{$\left(a\right)$ MSE versus $\text{SNR}_{\text{FC}}$ $ \left(b\right)$ MSE performance with varying uncertainty variance $\left(\sigma_H^2\right)$ of the hybrid beamforming design with stochastic CSI uncertainty.}
\label{MVDP2}
\end{figure*}
\section{Simulation Results}\label{sim_results}
In our simulations a mmWave MIMO IoTNe associated with ${M}=24$ sensors is considered, where each IoTNo is assumed to be equipped with ${N}_T=5$ antennas, and the FC employs ${N}_R=16$ antennas with  $N^m_{\text{RF}}=3$ and $N^{\text{FC}}_{\text{RF}}=4$.  Monte-Carlo simulations are performed for averaging the effects of the random channel realizations and other quantities. The mmWave MIMO channel $\mathbf{H}_m$ between each IoTNo $m$ and the FC is assumed to have ${N}_{cl}=10$ clusters and  each path-gain element ${\alpha}_{k,m}$ is generated as i.i.d. $ \mathcal{CN} \left(0,1\right)$ for $1 \leq k \leq N_{cl}$, and $1 \leq m \leq M$. The dimension of the unknown parameter vector is set as ${p}=4$ and the number of observations at each IoTNo is considered to be ${q}=5$. The correlation coefficient at each IoTNo $m$ is set as $\beta_m=0.6$, $1 \leq m \leq M$, while the correlation coefficient at the FC is set to $\alpha=0.6$. The elements of the observation matrix $\mathbf{A}_m$ are generated as $ \mathcal{CN} \left(0,1\right)$. The observation noise covariance matrix $\mathbf{R}_m$ for each IoTNo $m$ is set as $\mathbf{\sigma}_{m}^2 \mathbf{I}_{q}$, and the channel noise covariance matrix $\mathbf{R}_u$ is assumed to be $\mathbf{\sigma}_{u}^2 \mathbf{I}_{{N}_R}$. The observation SNR is defined as $\text{SNR}_{\text{OB}}=\frac{1}{\mathbf{\sigma}_{m}^2}$. Furthermore, the SNR at the FC is defined as $\text{SNR}_{\text{FC}}=\frac{1}{\mathbf{\sigma}_{u}^2}$, which is set separately for each experiment. \par


\begin{figure*}
\centering
\subfloat[]{\includegraphics[width=8.5cm, height=6.75 cm]{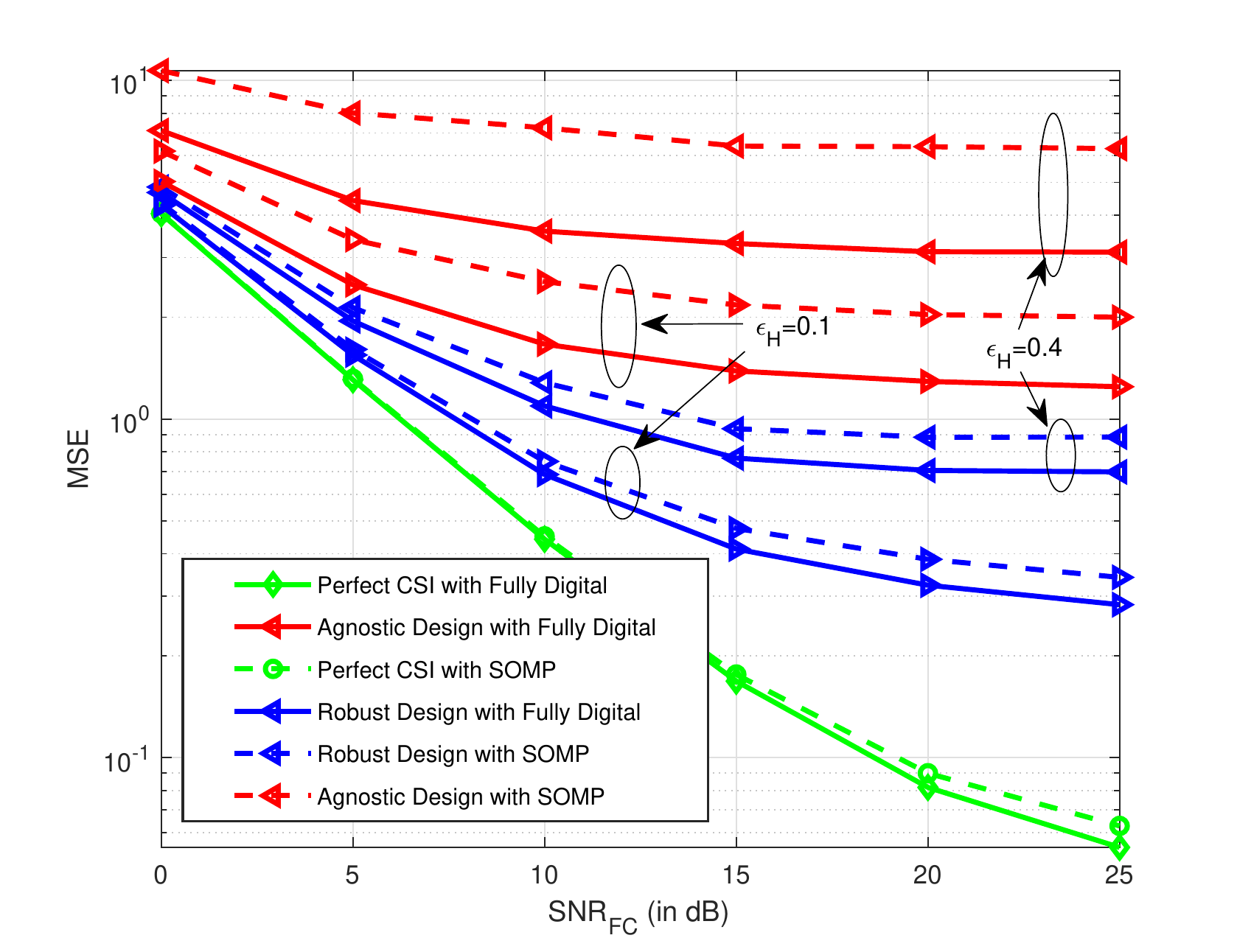}} 
\hfil
\hspace{-15pt}\subfloat[]{\includegraphics[width=8.5cm, height=6.75 cm]{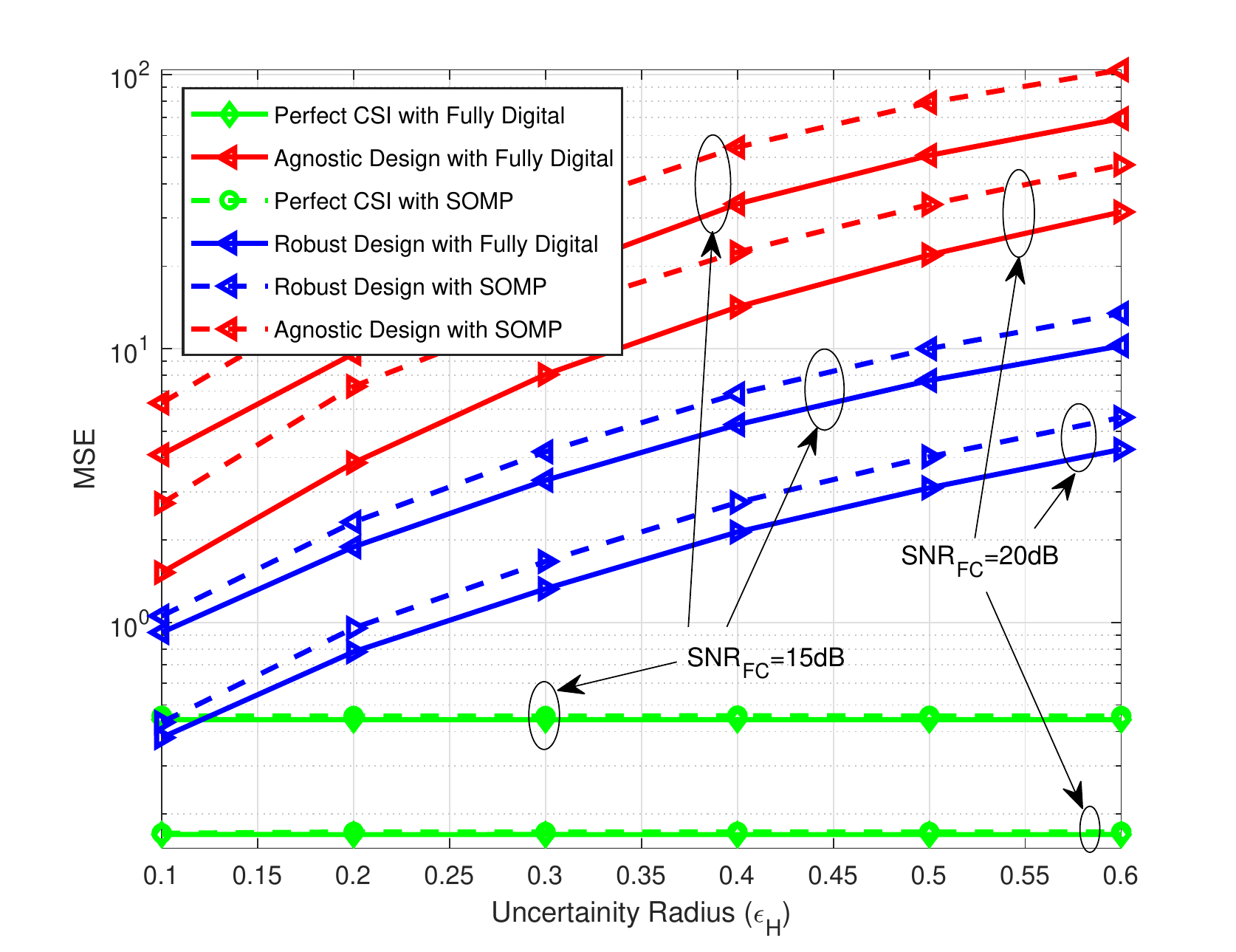}}
\hfil
\caption{$\left(a\right)$ MSE versus $\text{SNR}_{\text{FC}}$  $ \left(b\right)$ MSE versus channel uncertainty radius $\left(\epsilon_H\right)$ of the hybrid beamforming design with norm ball CSI uncertainty.}
\label{MVDP3}
\end{figure*}
Fig. \ref{MVDP} (a) shows the estimation performance of the three SOMP-based hybrid beamforming designs proposed in Section-III as a function of $\text{SNR}_{\text{FC}}$ for $M \in \{20,30\}$. The three schemes are 1) MSE minimization considering only the estimation constraint in \eqref{MSE_Min}, 2) estimation and total power constraint in \eqref{Total}, and finally 3) considering the estimation constraint and $M$ individual IoTNo power constraints as described in \eqref{individual}. Additionally, the corresponding optimal fully digital precoder's MSE performance is also plotted for benchmarking the performance of the SOMP-based hybrid designs. It can be readily observed from the figure that the SOMP based hybrid beamforming designs yield a performance close to the fully digital MMSE TPC designs, which bears testimony to the efficacy of the proposed designs. It also approaches to the centralized MMSE benchmark derived in \eqref{MMSE_Bench} at high $\text{SNR}_{\text{FC}}$. A theoretical proof for this is given in Appendix-\ref{AppB}. Also, the MSE performance further improves as the number of sensors increases, since more observations become available at the FC which results in an improved parameter estimation performance. \par
Fig. \ref{MVDP} (b) shows the MSE performance of the proposed hybrid TPC schemes as a function of the number of sensors $M$ in the IoTNe for different values of $\text{SNR}_{\text{OB}} \in \{-7, 1.25\}$ dB, with $\text{SNR}_{\text{FC}}$ set to $30$ dB. The MSE of the optimal fully digital TPC as well as the centralized MMSE benchmark are also plotted therein. Observe again that the MSE decreases as the number of sensors increases. Furthermore, increasing the values of $\text{SNR}_{\text{OB}}$ leads to a corresponding improvement in the estimation accuracy, thanks to the reduced observation noise variance. \par
\begin{figure*}
\centering
\subfloat[]{\includegraphics[width=8.5cm, height=6.75 cm]{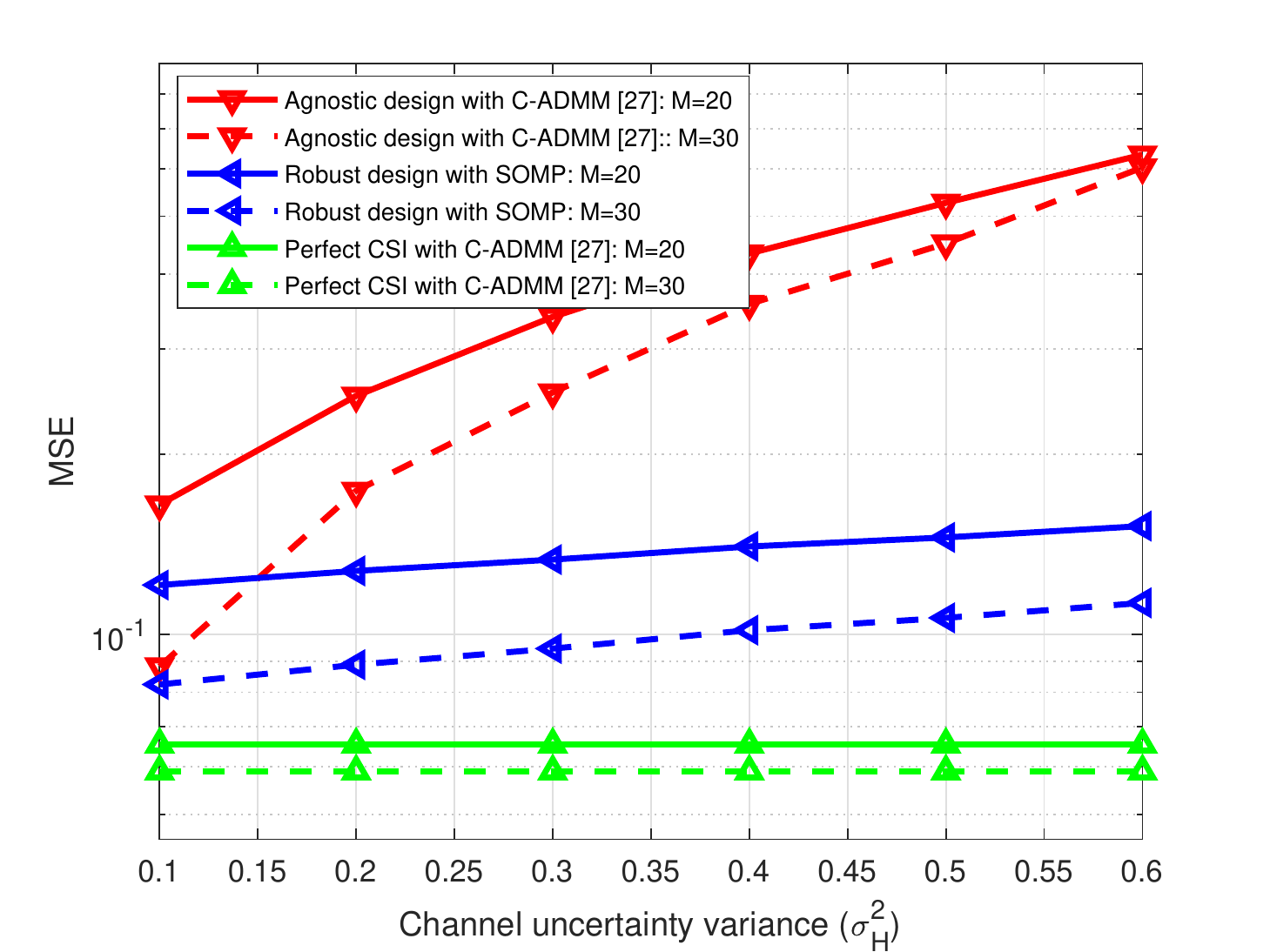}} 
\hfil
\hspace{-15pt}\subfloat[]{\includegraphics[width=8.5cm, height=6.75 cm]{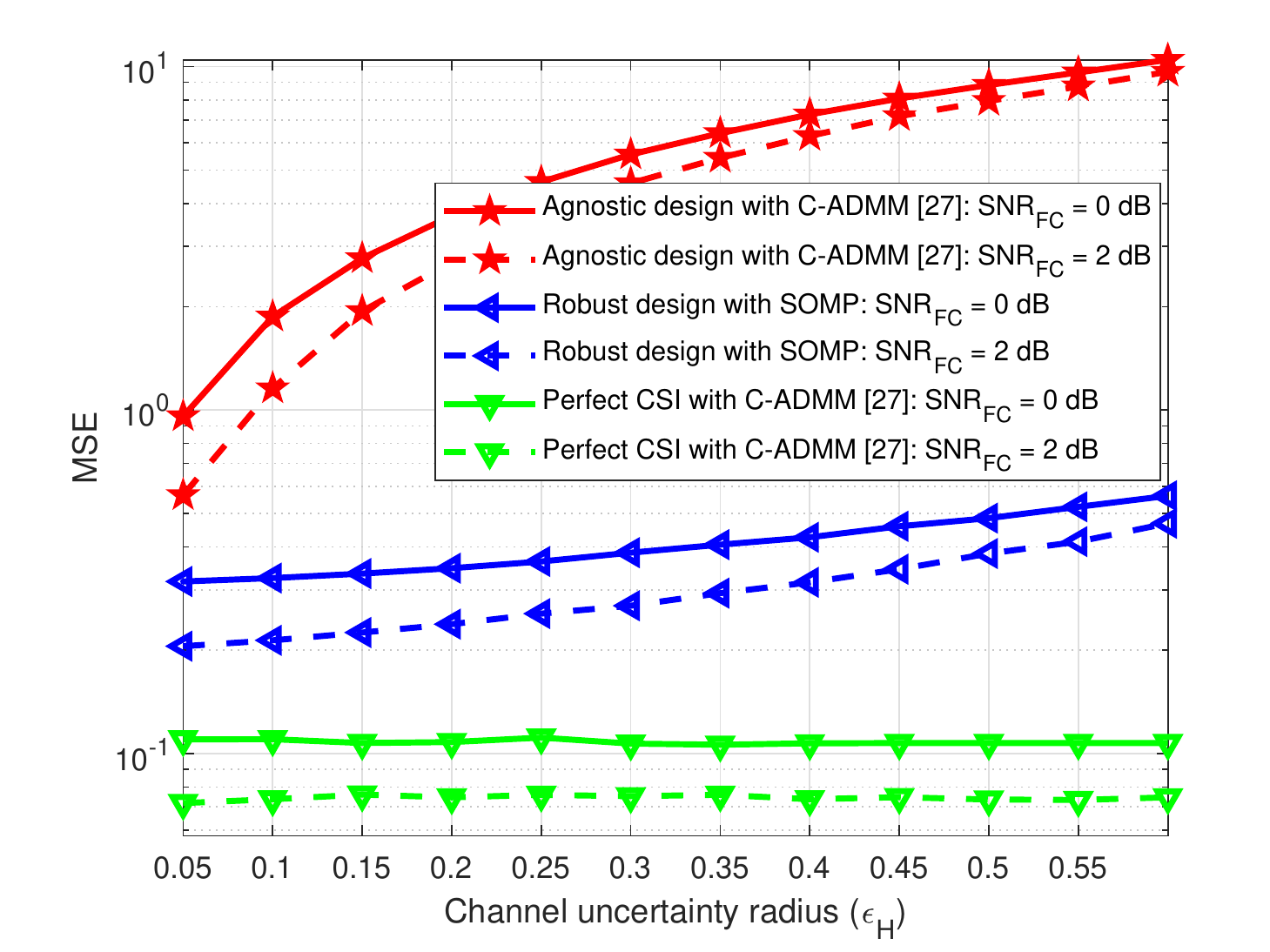}}
\hfil
\caption{\textcolor{black}{MSE performance with the C-ADMM scheme in \cite{liu2021hybrid} $\left(a\right)$ for the stochastic CSI uncertainty model with $\text{SNR}_{\text{FC}}=4\  \text{dB}, m=3, q=3, N_T = 10, N_R = 16, N_{\text{RF}}^m = 3, M \in \lbrace20,30\rbrace$ $ \left(b\right)$ for the norm ball CSI uncertainty model with $ M=10, m=3, q=3, N_T = 10, N_R = 16, N_{\text{RF}}^m = 3, \text{SNR}_{\text{FC}} \in \lbrace 0,2\rbrace \ \text{dB}$.}}
\label{comp}
\end{figure*}
Fig. \ref{fig:mse_snr_rf} depicts the MSE performance of the SOMP-based total-power-constrained design in Eq. (23) for the scenario having $N_{\text{RF}}^m \in \{1,2,3\}$ RF chains at each IoTNe. It can be readily observed from the figure that the proposed design is general in nature and can operate even with a single RF chain and multiple antennas at each IoTNe without significantly affecting the overall MSE performance of the system. \par
Fig. \ref{MVDP2} (a) characterizes the MSE performance of the robust hybrid beamforming design proposed in Section-\ref{Robust_Stochastic} for the stochastic CSI uncertainty model against varying $\text{SNR}_{\text{FC}}$ for the uncertainty variance $\sigma_H^2 \in \{0.1,0.4\}$. The performance of the corresponding average MMSE fully digital TPC is also plotted. Furthermore, in addition to the proposed robust design, the corresponding perfect and imperfect CSI-based designs are also included for illustrating the efficacy of the proposed robust design. The significant MSE performance gain achieved by the scheme presented over the uncertainty-agnostic counterpart that ignores the CSI uncertainty is clearly evidenced by the figure. Moreover, the MSE performance difference between the proposed robust and the uncertainty-agnostic design increases upon increasing the uncertainty variance $\sigma_H^2$. Fig. \ref{MVDP2} (b) explicitly shows the performance of the proposed robust hybrid beamforming design as a function of channel uncertainty variance $\sigma_H^2$. As expected, it can once again be observed that the MSE performance degrades as $\sigma_H^2$ increases. \par
Fig. \ref{MVDP3} (a) describes the worst-case MSE performance of the robust hybrid beamforming design proposed in Section-\ref{Robust_Bounded} for the norm ball CSI uncertainty model, and for different channel uncertainty radii $\epsilon_H \in \{0.1,0.4\}$.  The benefits of the presented  robust design are once again evidenced by its performance, which is close to that of the TPCs/ RC designed with perfect CSI and also by the  significant performance gain over the uncertainty-agnostic design. Reassuringly, the SOMP-based robust TPCs perform very close to their digital counterparts as well. This reinforces our claim that the robust designs formulated are capable of successfully overcoming the degradation arising due to the realistic imperfect CSI knowledge. Fig. \ref{MVDP3} (b) portrays the MSE performance as a function of the channel uncertainty radius $\epsilon_H$ for the perfect, robust and the agnostic designs  as well as for the corresponding fully digital baseband system for different values of $\text{SNR}_{\text{FC}} \in \{15,20\}$ dB. Similar to the previous figure, the robust design conceived for the norm ball CSI model once again outperforms the agnostic design with a fair margin, with the gap widening as $\text{SNR}_{\text{FC}}$ decreases.\par
\textcolor{black}{Fig. \ref{comp} (a) characterizes the MSE performance of the proposed SOMP-based robust hybrid TPC/ RC design along with the C-ADMM-based hybrid TPC/ RC design of \cite{liu2021hybrid} relying on both the perfect and with the CSI-uncertainty agnostic arrangements in the face of stochastic CSI uncertainty. The MSE is plotted as a function of the channel uncertainty variance $\sigma_H^2$ for different number of IoTNos, i.e., $M \in \{20,30\}$ subject to per IoTNo power constraints. It can be readily deduced from the figure that the proposed robust design offers a significant MSE performance improvement over the C-ADMM based uncertainty agnostic hybrid design. Furthermore, the MSE performance gap between both the robust and the CSI uncertainty agnostic design is seen to increase upon increasing $\sigma_H^2$. This illustrates the efficacy of our proposed robust design. Moreover, the MSE performance of the robust design improves upon increasing the number of IoTNos in mmWave MIMO IoTNe, reinforcing the trend seen earlier.}\par
\textcolor{black}{Fig. \ref{comp} (b) depicts the MSE performance of the proposed robust hybrid TPC/ RC design against the C-ADMM-based hybrid TPC/ RC design described in \cite{liu2021hybrid} for the scenario associated with norm ball CSI uncertainty. Once again, our proposed robust design is observed to achieve a significant MSE performance improvement over the C-ADMM based uncertainty agnostic design. Furthermore, with the increase in $\epsilon_H$, the MSE performance gap between the robust and agnostic designs is seen to increase. This illustrates that the proposed robust design is also efficient for the norm-ball CSI uncertainty model.}
\section{Conclusion}\label{Con}
Linear hybrid beamforming techniques were developed for vector parameter estimation in a coherent MAC-based mmWave MIMO IoTNe. Since the general MSE optimization is non-convex due to the constant gain constraint pertaining to the RF phase-shifters that form the RF TPC and RC, the optimal minimum MSE digital TPC is initially obtained under the zero-forcing and total/  per IoTNo power constraints. Subsequently, the popular SOMP algorithm was employed for decomposing the fully digital TPC into its RF and baseband components. Next, novel average and worst-case MSE analyses were derived for designing the robust TPCs for the stochastic and norm ball models, respectively, in scenarios associated with CSI uncertainty. The centralized MMSE that acts as a lower bound was determined for the system, which yields exceptional insights into the MSE performance.  Detailed simulation results demonstrated the efficacy of the proposed algorithms and a performance close to the MMSE bound.


\appendices
\section{Proof of Lemma 1}\label{Lemma1}
Since the matrix $\mathbf{\Psi}_m$ is defined as $\mathbf{\Psi}_m= \left(  \mathbf{R}_m \otimes \mathbf{H}_m^H \mathbf{H}_m \right)$, this implies that \\ $\mathbf{\Psi}_m^{-1} = \left[  \mathbf{R}_m^{-1} \otimes \left( \mathbf{H}_m^H \mathbf{H}_m \right)^{-1} \right]$, which can be further expressed as
\begin{align}\label{EqLemma}
\mathbf{\Psi}_m^{-1} =&\begin{bmatrix}
r_{11}\left( \mathbf{H}_m^H \mathbf{H}_m \right)^{-1} & \cdots & r_{1m}\left( \mathbf{H}_m^H \mathbf{H}_m \right)^{-1}\\
r_{21}\left( \mathbf{H}_m^H \mathbf{H}_m \right)^{-1}  &  \cdots & r_{1m}\left( \mathbf{H}_m^H \mathbf{H}_m \right)^{-1}\\
\vdots & \ddots & \vdots\\
r_{m1}\left( \mathbf{H}_m^H \mathbf{H}_m \right)^{-1}  & \cdots & r_{mm}\left( \mathbf{H}_m^H \mathbf{H}_m \right)^{-1}
\end{bmatrix} ,   
\end{align}
where $r_{ij} = \left[\mathbf{R}_m^{-1}\right]_{i,j}$. From the expression in \eqref{EqLemma}, it can be seen that any column of the matrix $\mathbf{F}_m$ can be written as a linear combination of the columns of the matrix $\left( \mathbf{H}_m^H \mathbf{H}_m \right)^{-1}$. Let the singular value decomposition (SVD) of $\mathbf{H}_m$ be defined as  $\mathbf{H}_m= \mathbf{U}_m \mathbf{\Sigma}_m \mathbf{V}_m^H$. Thus, we have $\left( \mathbf{H}_m^H \mathbf{H}_m \right)^{-1} = \mathbf{V}_m \tilde{\mathbf{\Sigma}}_m \mathbf{V}_m^H$, where $\tilde{\mathbf{\Sigma}}_m = \mathrm{diag} \left\{ \big[ \mathbf{\Sigma}_m(i,i) \big]^{-2} \right\}_{i=1}^{t}$, which implies that the column space of the TPC $\mathbf{F}_m$ lies in the column space of the matrix $\mathbf{V}_m$. Finally, one can conclude that
\begin{align}
\mathcal{C}(\mathbf{F}_m)\subseteq \mathcal{R}(\mathbf{H}_m),
\end{align} 
since, the column space of $\mathbf{V}_m$ constitutes the row space of the channel matrix $\mathbf{H}_m$.
\section{Asymptotic MSE analysis for the MSE minimization proposed in \eqref{MSE_Min}}\label{AppB}
\begin{thm}
Asymptotically, i.e., at very high $\text{SNR}_{\text{FC}}$, the MSE of the linear hybrid TPC design developed in \eqref{MSE_Min} approaches the centralized MMSE benchmark derived in \eqref{MMSE_Bench}, i.e.,
\begin{equation}
\lim_{\mathrm{SNR}_{\mathrm{FC}} \to \infty} \mathrm{MSE}-\mathrm{MSE}_{\mathrm{MMSE}}=0.
\end{equation}
\end{thm} 
\begin{proof}
Exploiting the expression of $\mathrm{MSE}_{\mathrm{MMSE}}$ derived in \eqref{MMSE_Bench}, and substituting the eigenvalue decomposition of the matrix given by $\mathbf{A}^H \mathbf{R}_{v}^{-1} \mathbf{A}=\mathbf{Q}\mathbf{\Lambda}\mathbf{Q}^H$, one obtains
\begin{equation}
\mathrm{MSE}_{\mathrm{MMSE}}=\sum_{m = 1}^{p} \frac{1}{1+\lambda_m \left(\mathbf{A}^H \mathbf{R}_{v}^{-1}\mathbf{A}\right)},
\end{equation}
where without loss of generality $\mathbf{R}_{\theta}=\mathbf{I}_p$ and $\lambda_m(\mathbf{X})$ denotes the $m$th eigenvalue of the matrix $\mathbf{X}$. Substituting the optimal TPC vector expression derived in \eqref{Opt_Precoder} into the MSE expression of \eqref{MSE1}, one obtains
\begin{equation}\label{MSE_final}
\mathrm{MSE}=\mathbf{f}^H\left(\mathbf{Z}\mathbf{\Psi}^{-1}\mathbf{Z}^H\right)\mathbf{f} +\mathrm{Tr}\left[\mathbf{W}_{\text{RF}}^H \mathbf{R}_u \mathbf{W}_{\text{RF}}\right].
\end{equation}
The quantity ${\mathbf{Z}\mathbf{\Psi}^{-1}\mathbf{Z}^H}$ can be further simplified as
\begin{align}
&{\mathbf{Z}\mathbf{\Psi}^{-1}\mathbf{Z}^H}\nonumber \\
&= \sum_{m = 1}^{p} \left[\mathbf{A}_m^T \otimes \mathbf{W}_{\text{RF}}^H \mathbf{H}_m \right]{\left[\mathbf{R}_{m}^T \otimes \mathbf{H}_m^H \mathbf{W}_{\text{RF}} \mathbf{W}_{\text{RF}}^H \mathbf{H}_m \right]}^{-1} \nonumber \\ &\hspace{+30pt}\left[\mathbf{A}_m^\ast \otimes \mathbf{H}_m^H \mathbf{W}_{\text{RF}} \right] \nonumber \\
&= \sum_{m = 1}^{p} \Big[\left[\mathbf{A}_m^T \mathbf{R}_{m}^{-T}\mathbf{A}_m^\ast \right] \otimes \Big[\mathbf{W}_{\text{RF}}^H \mathbf{H}_m {\left(\mathbf{H}_m^H \mathbf{W}_{\text{RF}} \mathbf{W}_{\text{RF}}^H \mathbf{H}_m \right)}^{-1} \nonumber \\ &\hspace{+30pt}\mathbf{H}_m^H \mathbf{W}_{\text{RF}} \Big] \Big] \nonumber \\
&= \sum_{m = 1}^{p} \left[\mathbf{A}_m^T \mathbf{R}_{m}^{-T}\mathbf{A}_m^\ast \otimes \mathbf{I}_p \right],  
\end{align}
where the simplification above exploits the properties $\left(\mathbf{X} \otimes \mathbf{Y}\right)\left(\mathbf{Z} \otimes \mathbf{W}\right)=\left(\mathbf{X}\mathbf{Z} \otimes \mathbf{Y}\mathbf{W}\right)$ and $\left[\mathbf{X}_1 \otimes \mathbf{Y} , \mathbf{X}_2 \otimes \mathbf{Y}\right]=\left[\mathbf{X}_1, \mathbf{X}_2\right] \otimes \mathbf{Y}$. Now, substituting the above expression for $\mathbf{Z}\mathbf{\Psi}^{-1}\mathbf{Z}^H$ together with the property in \eqref{Trace_vec}, the first term of the MSE expression in \eqref{MSE_final} can be simplified as
\begin{align}
\mathbf{f}^H\left(\mathbf{Z}\mathbf{\Psi}^{-1}\mathbf{Z}^H\right)\mathbf{f}&=\mathbf{f}^H\left[\mathbf{A}^T \mathbf{R}_{v}^{-T}\mathbf{A}^\ast \otimes \mathbf{I}_p \right]^{-1}\mathbf{f}\nonumber \\
&=\mathrm{vec} \left[\mathbf{I}_p^H \right]\left[\left[\mathbf{A}^T \mathbf{R}_{v}^{-T}\mathbf{A}^\ast \right]^{-1} \otimes \mathbf{I}_p \right]\mathrm{vec} \left[\mathbf{I}_p \right]\nonumber \\
&=\mathrm{Tr}\left[\left[\mathbf{A}^H \mathbf{R}_{v}^{-1}\mathbf{A}\right]^{-1}\right]\nonumber \\
&=\sum_{m = 1}^{p} \frac{1}{\lambda_m \left(\mathbf{A}^H \mathbf{R}_{v}^{-1}\mathbf{A}\right)}.
\end{align}
At high $\text{SNR}_{\text{FC}}$, $\mathrm{Tr}\left[\mathbf{W}_{\text{RF}}^H \mathbf{R}_u \mathbf{W}_{\text{RF}}\right]=0$, and with high SNR IoTNo observations, we have $\lambda_m \left(\mathbf{A}^H \mathbf{R}_{v}^{-1}\mathbf{A}\right) >> 1$. Therefore, one can conclude that
\begin{align}
&\lim_{\mathrm{SNR}_{\mathrm{FC}} \to +\infty} \mathrm{MSE}-\mathrm{MSE}_{\mathrm{MMSE}}\nonumber \\&=\sum_{m = 1}^{p} \frac{1}{\lambda_m \left(\mathbf{A}^H \mathbf{R}_{v}^{-1}\mathbf{A}\right)}-\sum_{m = 1}^{p} \frac{1}{1+\lambda_m \left(\mathbf{A}^H \mathbf{R}_{v}^{-1} \mathbf{A}\right)}\nonumber \\
&=\sum_{m = 1}^{p} \frac{1}{\lambda_m \left(\mathbf{A}^H \mathbf{R}_{v}^{-1}\mathbf{A}\right)}-\sum_{m = 1}^{p} \frac{1}{\lambda_m \left(\mathbf{A}^H \mathbf{R}_{v}^{-1} \mathbf{A}\right)}\nonumber \\
&=0.\nonumber
\end{align}

\end{proof}
 
\bibliographystyle{IEEEtran}
\bibliography{reference}
\end{document}